\documentclass{amsart}
\usepackage{amsfonts, amsmath, amsthm, physics, graphicx, stmaryrd, cite, enumerate, hyperref}
\usepackage{adjustbox}
\usepackage[all]{xy}
\usepackage[foot]{amsaddr}
\allowdisplaybreaks

\newtheorem{theorem}{Theorem}

\newtheorem{lemma}{Lemma}
\newtheorem*{lemma*}{Lemma}

\newcommand{\bd}{\partial} 
\newcommand{\ld}{\lambda}
\newcommand{\Ld}{\Lambda}
\newcommand{\bs}{\backslash}
\newcommand{\lb}{\langle}
\newcommand{\rb}{\rangle}

\newcommand\numberthis{\addtocounter{equation}{1}\tag{\theequation}}

\begin{document}
 
\title[Spectral Gaps of the Two-Species PVBS Models]{Spectral gaps for the Two-Species Product Vacua and Boundary States models on the $d$-dimensional lattice}
\author{Michael Bishop}
\date{\today}
\email{mibishop@csufresno.edu}
\address{Department of Mathematics, California State University, Fresno, Fresno, CA, USA}

\maketitle
\begin{abstract}
We study the two-species Product Vacua and Boundary States (PVBS) models on the integer lattice $\mathbb{Z}^d$ and prove the existence and non-existence of a spectral gap for all choices of parameters.  
The PVBS models are spin-1 quantum spin systems which are translation-invariant, frustration-free, and composed of nearest-neighbor non-commuting interactions with both an exclusion property and an interchange interaction between particle species.  
These models serve as possible representatives of families of automorphically equivalent gapped quantum spin-1 systems on $\mathbb{Z}^d$.
The main result is that the two-species PVBS Hamiltonians have a positive spectral gap when gapped on both of the single-species subspaces and are gapless if gapless on either single-species subspace.  
The addition of a new particle species does not create any new gapless phases.
\end{abstract}

\section{Introduction}

The existence of a spectral gap is a crucial property in describing the behavior of low temperature quantum spin systems.   
The spectral gap is the distance between the bottom and the rest of the spectrum of the Hamiltonian operator in the description the system: if the distance is positive, we say the system is gapped; if it is zero, we say the system is gapless. 
Transitions between gapped and gapless cases due to changes in the parameters of the model are called `quantum phase transitions' because they indicate a stark change in the behavior of the system.
If there is a positive spectral gap above the ground state space, techniques developed in \cite{PhysRevB.69.104431,  Nachtergaele2006, Aharonov:2009:DLQ:1536414.1536472, aharonov2011detectability, PhysRevB.93.205142} establish that these systems exhibit exponential decay of correlations \cite{Hastings2006, PhysRevLett.116.097202}.  
These systems also exhibit area laws for their entropy in one dimension \cite{Hastings07, arad2013area} among other properties.  
The existence of a uniform spectral gap along the required path of Hamiltonians is a sufficient condition for quantum computation using adiabatic evolution \cite{farhi2000quantum, schwarz2013preparing, jansen2007bounds}.
If there is no spectral gap above the ground state, then many of these properties may fail \cite{gottesman2010entanglement, irani2010ground,  PhysRevLett.109.207202, movassagh2014power}.

Methods for determining and bounding spectral gaps from below for a quantum spin system are largely inspired by the AKLT model \cite{aklt88}.  
The method introduced in \cite{knabe1988energy} and improved in \cite{gosset2016local} bounds the spectral gap for a spin chain with periodic boundary conditions when the spectral gaps on subsystems satisfy certain bounds.  
Recently, this method has been extended to higher dimensional systems with boundary \cite{lemm2018spectral}.  
In contrast, the key condition to apply the `martingale method' \cite{nachtergaele96} depends on the finite volume ground state projections rather than a condition on local spectral gaps. 
Recently, the method of \cite{nachtergaele96} has been better adapted for higher dimensional systems \cite{kastoryano2017divide}.  
Both methods were used to classify the gapped and gapless phases for a family of spin-1/2 systems on spin chains of $\mathbb{Z}$ \cite{2015JMP....56f1902B}.  
 
For systems on $\mathbb{Z}^d$, there cannot exist a general criterion for proving the existence of a spectral gap from the description of the system on finite volumes \cite{cubitt2015undecidability}.  
That is, the existence of a spectral gap is `undecidable' from the local descriptions and may unexpectedly open or close for a given family of local interactions along a given sequence of finite volumes.  
Gapped systems require careful construction: Hamiltonians with randomly distributed local interactions are almost surely gapless under mild assumptions \cite{movassagh2017generic}.

Under these constraints, a natural goal is to find families of models for which we can determine the existence of a spectral gap and find methods to connect them by automorphic equivalence \cite{BMNS2012} to larger families of quantum spin systems.  
Simplified models such as the toric code model\cite{kitaev2003fault}, deformations of symmetry-protected systems \cite{PhysRevB.83.035107,PhysRevB.87.155114}, and systems described using matrix product states (MPS) and product entangled pair states (PEPS) \cite{PhysRevB.84.165139} have proven fruitful in exploring gapped quantum systems and shedding light on unexpected properties including topological order. 

The Product Vacua and Boundary States (PVBS) models are a family of quantum spin systems which are translation-invariant, frustration-free, and composed of nearest-neighbor non-commuting interactions with an exactly formulated ground state space.   
In \cite{bachmann2014product}, the authors described the gapped and gapless phases of the PVBS model on the infinite chain and half-chains using a Matrix Product State (MPS) representation.
They demonstrated that a subfamily of these models were automorphically equivalent to the AKLT model and conjectured the PVBS models could serve as representatives for classifying equivalence classes of gapped quantum spin systems.
The PVBS models on subsets of $\mathbb{Z}^d$ were introduced in \cite{BHNY2015} as a step towards potentially proving spectral gaps for higher dimensional AKLT and related models.
The spectral properties in the single-species (spin-1/2) version were explored using the martingale method. 
These models had the peculiar property due to higher dimensional support on $\mathbb{Z}^d$: PVBS models which were gapped on $\mathbb{Z}^d$ may be gapless when restricted to a specific half-space.
In \cite{bishop2016spectral}, this discrepancy was resolved and the gapped and gapless ground states of the single-species PVBS model were determined on $\mathbb{Z}^d$ and its half-spaces.   

The goal of this paper is to extend the results of \cite{bishop2016spectral} to understand the spectral properties of the two-species (spin-1) PVBS model on $\mathbb{Z}^d$ and explore the possibility of new gapless phases due to interspecies interactions.  
These models can be viewed as families of interacting particle systems of two distinct species (with exclusion) occupying sites on the lattice.   
The main result of this paper is that a two-species PVBS Hamiltonian on $\mathbb{Z}^d$ is gapped if and only if it is gapped on both single-species subspaces. 
The proof applies the martingale method \cite{nachtergaele96} to appropriate sequences of volumes; the main technical work is verifying that the ground state spaces satisfy the conditions of the method.  
This result demonstrates that the introduction of a new particle species does not create new gapless states with both particle species present. 
While this suggests that the two-species model is embedded in an independent product of two single-species models, this embedding is not readily apparent. 
Conversely, if the Hamiltonian on $\mathbb{Z}^d$ restricted to either single species subspace is gapless, then the two-species model is gapless as well because the system restricted to either species is equivalent to a single species PVBS model.  
This suggests that a similar result should hold for PVBS models of any finite number of species and thus any spin number.
 
The paper is structured as follows.  
The model and ground state space are introduced in section 2.  
The main result is stated and proven in section 3.
The finite volume ground state space proofs are in section 4.
The technical calculations needed to apply the martingale method are in section 5.   
 
\section{Product Vacua with Boundary State Models and Ground State Spaces} 

\subsection{Finite Volume Model}
Let $\Lambda$ be a bounded subset of $\mathbb{Z}^d$. 
At each site $x\in\Ld$, we define the single-site Hilbert space $\mathbb{C}_x^3$ with orthonormal basis $\{\ket{0}_x,\ket{a}_x,\ket{b}_x\}$.  
These vectors correspond to the site $x$ being occupied by no particle, a particle of species $a$, and a particle of species $b$, respectively.  
The species exhibit an exclusion property where only one particle may occupy a given site. 
We define the finite volume Hilbert space to be $\mathcal{H}^\Ld :=\otimes_{x\in\Ld} \mathbb{C}^3$.  
An orthonormal basis can be indexed by disjoint subsets $A$ and $B$ of $\Ld$ where corresponding orthonormal vectors have the form
\begin{equation}
\bigotimes\limits_{x\in A} \ket{a}_x \otimes \bigotimes\limits_{y \in B} \ket{b}_y \otimes \bigotimes\limits_{z \in \Ld \bs (A\cup B)} \ket{0}_z \label{ONB}
\end{equation}
where $A$ is the set the sites occupied by a particle of species $a$ and $B$ is the set the sites occupied by a particle of species $b$.  
For all other sites, the site is unoccupied; to abbreviate notation the vacuum state $\bigotimes \ket{0}$ is dropped.  
The subspace with $|B|=0$ is the $a$-particle subspace and is unitarily-equivalent to a single-species PVBS Hilbert space on $\Ld$; the $b$-particle subspace is similarly defined by $|A|=0$.     

The two-species PVBS Hamiltonians are indexed by the parameter vectors 
\begin{equation}
\ld_a=(\ld_{a,1}, \ld_{a,2}, \dots, \ld_{a,d}),\ \ld_b=(\ld_{b,1}, \ld_{b,2}, \dots \ld_{b,d})\label{paramvec}
\end{equation}
where the subscripts denote the associated particle species and the coordinate directions on $\mathbb{Z}^d$, respectively.  
Without loss of generality, we may assume the $\ld_{s,j}$ are strictly positive for $s=a,b$ and for $j=1,\dots,d$.  
In the case that the parameters are non-zero complex numbers, the model is unitarily equivalent to a model with positive parameters equal to the modulus of the complex parameter.  

We impart a graph structure to $\mathbb{Z}^d$ by connecting each pair of sites $x$ and $x+e_j$ by an edge where $e_j$ is the standard basis vector in $j$-th coordinate direction of $\mathbb{Z}^d$.
To each edge $x,\ x+e_j$, we define the following edge projection on the subspace of $\mathbb{C}_x^3\otimes\mathbb{C}_{x+e_j}^3$, spanned by
$$ \ket{0,a} - \ld_{a,j}\ket{a,0},\  
\ket{0,b} - \ld_{b,j}\ket{b,0},\  
\ld_{a,j}\ket{a,b} - \ld_{b,j}\ket{b,a},\ \ket{a,a},\ \ket{b,b}.$$
Equivalently, 
\begin{align*}
	h_{x,x+e_j} :=& \frac{\left(\ket{0,a} - \ld_{a,j}\ket{a,0}\right)\left( \bra{0,a} - \ld_{a,j}\bra{a,0}\right)}{1+\ld_{a,j}^2} \\
	&+ \frac{\left(\ket{0,b} - \ld_{b,j}\ket{b,0}\right)\left( \bra{0,b} - \ld_{b,j}\bra{b,0}\right)}{1+\ld_{b,j}^2} \\
	&+ \frac{\left(\ld_{a,j}\ket{a,b} - \ld_{b,j}\ket{b,a}\right)\left(\ld_{a,j} \bra{a,b} - \ld_{b,j}\bra{b,a}\right)}{\ld_{a,j}^2 + \ld_{b,j}^2}\\
	&+ \dyad{a,a} + \dyad{b,b} \numberthis\label{edgeproj}
\end{align*}
where $\ket{s, s'}$ denotes $\ket{s}_x\otimes\ket{s'}_{x+e_j} \in \mathbb{C}_x^3\otimes\mathbb{C}_{x+e_j}^3$. 
The first and second terms effectively `hop' the $a$ particles and $b$ particles, respectively; the third term interchanges neighboring particles of different species; the final two terms repel particles of the same species.   
These edge projections preserve the number and species of particles occupying either $x$ or $x+e_j$.  
The finite-volume Hamiltonian on $\Ld$ is the sum of these edge projections
\begin{align} \label{fvham}
	H^\Ld := \sum\limits_{j=1}^d\ \sum\limits_{x,x+e_j \in \Ld} h_{x,x+e_j}
\end{align}
where each edge projection is extended by the identity operator on $\mathcal{H}^{\Ld\bs\{x,x+e_j\}}$.  
This Hamiltonian preserves the number of each species of particle in $\mathcal{H}^\Ld$.  
The restriction of this Hamiltonian to the set of vectors with only one particle species $s$ is equivalent to the single-species PVBS Hamiltonian with $\log\ld = \log\ld_s$.

The ground state space of $h_{x,x+e_j}$ in $\mathbb{C}_x^3\otimes\mathbb{C}_{x+e_j}^3$ is four dimensional and spanned by
\begin{equation}
\ket{0,0},\ 
\ld_{a,j}\ket{0,a} +\ket{a,0},\  
\ld_{b,j}\ket{0,b} + \ket{b,0},\  
\ld_{b,j}\ket{a,b} + \ld_{a,j}\ket{b,a}\label{edgegs}.
\end{equation}
The finite volume PVBS Hamiltonian ground state space is spanned by four corresponding vectors.
Each  are in the ground state space for all edge projections $h_{x,x+e_j}$ and therefore, $H^\Ld$ is \textit{frustration-free}.    

\begin{theorem}[Two-species Finite Volume Ground State Space of $H^\Ld$]\label{fvgsthem}
For finite and connected $\Ld \subset \mathbb{Z}^d$ and any choice of positive parameters $\ld_s(j)$, $s =a,b;\ j=1, \dots, d$, the ground state space of the associated PVBS Hamiltonian $H^\Ld$ is four-dimensional and spanned by the orthonormal set:
\begin{align*}\label{fvgsv}
&\Psi_0^\Lambda := \otimes_{x\in\Lambda} |0\rangle_x\\
&\Psi_a^\Lambda := \frac{1}{\sqrt{C(\Lambda,a)}}\sum\limits_{x\in\Lambda} \lambda_a^x |a\rangle_x\\
&\Psi_b^\Lambda: = \frac{1}{\sqrt{C(\Lambda,b)}}\sum\limits_{y\in\Lambda} \lambda_b^x |b\rangle_y\\
&\Psi_{ab}^\Lambda := \frac{1}{\sqrt{C(\Lambda,ab)}}\sum\limits_{\substack{x,y\in\Lambda\\ x\neq y}} \lambda_a^x\lambda_b^y |a\rangle_x\otimes|b\rangle_y \numberthis
\end{align*}
where
\begin{align*}
\lambda_s^x :=\prod\limits_{j=1}^d \lambda_{s,j}^{x_j} \text{ for } s=a,b \numberthis\label{ldprod}\\
C(\Lambda, s) := \sum\limits_{x\in\Lambda} \lambda_s^{2x} \text{ for } s = a, b \numberthis \label{CLdt}\\
C(\Lambda,ab) := \sum\limits_{\substack{x,y\in\Lambda\\ x\neq y}} \lambda_a^{2x}\lambda_b^{2y} \numberthis\label{CLdab}\\
\end{align*}
and $|s\rangle_x$ for $s = a, b$ and $ |a\rangle_x\otimes|b\rangle_y$ are extended by $\bigotimes \ket{0}$ to the rest of $\Ld$.  
\end{theorem}
The single-species ground state vectors $\Psi_a$ and $\Psi_b$ are exactly the ground states for the single-species PVBS Hamiltonian of the corresponding particle species.  
The two-species ground state vector is not simply a product of the two single-species ground states due to the exclusion property between the particle species.  
Consequently, the normalization constant for the $ab$ ground state is not the product of the normalization constants for the ground states for each particle species.  
It is bounded above by the product, $C(\Lambda, ab) < C(\Lambda,a)C(\Lambda,b)$, where the difference is exactly the diagonal term 
\begin{equation}\label{DLd}
 D(\Ld) := C(\Ld,a)C(\Ld,b)-C(\Ld, ab) = \sum_{x\in\Lambda}\lambda_a^{2x}\lambda_b^{2x}.
\end{equation}  
The proof of Theorem 1 follows a similar argument as the proof from proposition 2.1 in \cite{BHNY2015} and is presented in section \ref{groundstates}.  

For both parameter vectors associated to species $a$ and $b$, we define the logarithm vector of them coordinate-wise:
\begin{align}\label{logld}
	&\log\ld_s := (\log\ld_{s,1}, \log\ld_{s,2}, \dots, \log\ld_{s,d})
\end{align}  
so that $ \ld_s^{x}=\exp(x\cdot \log\ld_s)$.  
This will be used simplify bounds for separate cases by the same expression.  
Heuristically, the probability of finding the particle species $s$ at site $x$ in the associated (non-vacuum) finite volume ground states is proportional to $\exp(2x\cdot \log\ld_s)$.
We may interpret $\log\ld_s$ as the energetically favored direction for particle type $s$.

\subsection{Infinite Volume Model}
The infinite volume PVBS model is defined on connected infinite subsets $\Gamma \subseteq \mathbb{Z}^d$ using the Gelfand-Naimark-Segal (GNS) construction.  
For increasing and absorbing sequences of finite volumes $\Ld_n$ which converge to $\Gamma$, define $\mathcal{A}^{\Ld_n}$ to be the algebra of bounded operators on $\mathcal{H}^{\Ld_n}$.  
These algebras may be embedded into algebras on larger finite volumes by extending elements in the algebra on the smaller volume by identity on the rest of the larger volume.  
The algebra of local observables on $\Gamma$ is denoted $\mathcal{A}^{\Gamma}_{loc}$ and defined as
\begin{equation}
	\mathcal{A}^{\Gamma}_{loc} = \bigcup\limits_{n} \mathcal{A}^{\Ld_n}.
\end{equation}  
The operator norm is well-defined on this algebra; we define the algebra of quasi-local observables $\overline{\mathcal{A}^{\Gamma}_{loc}}$ as the norm-closure of the local observables on $\Gamma$. 

The infinite volume ground states are easily classified for the infinite volume two-species PVBS models.
The finite volume vacuum states are defined by
$\omega^{\Ld_n}_0 (\cdot):=  \lb \Psi^{\Ld_n}_0, \cdot  \Psi^{\Ld_n}_0 \rb $ 
acting on the algebra of local observables on $\mathcal{A}^{\Ld_n}$.  
These finite volume vacuum states weakly converge to the unique infinite volume vacuum state $\omega^\Gamma_0(\cdot)$ on $\overline{\mathcal{A}^{\Gamma}_{loc}}$ independent of the choice of increasing and absorbing sequence $\Ld_n$.  
With respect to its Gelfand-Naimark-Segal(GNS) representation $(\pi_0^\Gamma,\mathcal{H}_0^\Gamma, \Omega_0^\Gamma)$, $\omega^\Gamma _0(\cdot) = \lb \Omega^\Gamma _0, \pi^\Gamma_0(\cdot) \Omega^\Gamma_0\rb$, where $\pi^\Gamma_0$ is the representation of $\overline{\mathcal{A}^{\Gamma}_{loc}}$ in the bounded operators of $\mathcal{H}_0^\Gamma$. 
There may be other infinite volume ground states depending on whether $C(\Ld_n, s)$ converge or diverge when $\Ld_n \nearrow \Gamma$.  
The other finite volume ground states are 
\begin{align}
	\omega_a^{\Ld_n} (\cdot ) = \lb \Psi_a^{\Ld_n}, \cdot\ \Psi_a^{\Ld_n} \rb\\
	\omega_b^{\Ld_n} (\cdot ) = \lb \Psi_b^{\Ld_n}, \cdot\ \Psi_b^{\Ld_n} \rb\\
	\omega_{ab}^{\Ld_n} (\cdot ) = \lb \Psi_{ab}^{\Ld_n}, \cdot\ \Psi_{ab}^{\Ld_n} \rb
\end{align}
Let $p^a_x$ be the linear operator on $\mathbb{C}^3_x$ that maps $\ket{0}_x$ to $\ket{a}_x$, $\ket{a}_x$ to $\ket{0}_x$, and $\ket{b}_x$ to itself and let $p^b_x$ be defined similarly by interchanging $a$ and $b$.
These maps are extended by identity on $\mathcal{H}^{\Ld_n\bs\{x\}}$.  
If all the normalization constants converge in the infinite volume limit,  
there are non-vacuum ground states in GNS Hilbert space $\mathcal{H}^\Gamma_0$:
\begin{align} 
&\Omega_a^{\Gamma} = \frac{1}{\sqrt{C(\Gamma,a)}}\sum\limits_{x\in\Gamma} \lambda_a^x \pi^\Gamma_0(p_x^a) \Omega^\Gamma_0 \label{ivqlgs1}\\
&\Omega_b^{\Gamma}: = \frac{1}{\sqrt{C({\Gamma},b)}}\sum\limits_{y\in{\Gamma}} \lambda_b^x \pi^\Gamma_0(p_x^b) \Omega^\Gamma_0\label{ivqlgs2}\\
&\Omega_{ab}^{\Gamma} := \frac{1}{\sqrt{C({\Gamma},ab)}}\sum\limits_{\substack{x,y\in{\Gamma}\\ y\neq x}} \lambda_a^x\lambda_b^y \pi^\Gamma_0(p_x^a p_y^b) \Omega^\Gamma_0\label{ivqlgs3}
1\end{align}
where each of these are well-defined in the infinite volume vacuum GNS space and are orthogonal to each other and the infinite vacuum vector.  
Each of these states is in the kernel every edge projection $\pi^\Gamma_0 (h_{x,x+e_j})$ and thus the PVBS Hamiltonians on the infinite volume are frustration-free.  
The convergence of finite volume ground states to either its corresponding infinite volume ground state or the infinite volume vacuum state state is summarized by the following theorem.

\begin{theorem}[Infinite Volume Ground States]\label{ivgs}
Let $\Gamma$ be an connected infinite lattice and let $\Ld_n$ be a sequence of increasing and absorbing finite connected volumes converging to $\Gamma$. In the weak-$*$ topology,
\begin{enumerate}[i.]
	\item For either $s=a,b$, if $\lim_{n\to\infty}C(\Ld_n,s) = +\infty$, then $\omega^{\Ld_n}_s(\cdot) \to \omega^\Gamma_0(\pi^\Gamma_0(\cdot))$ and $\omega^{\Ld_n}_{ab}(\cdot) \to \omega^\Gamma_0(\pi^\Gamma_0(\cdot))$.\\
	\item For either $s=a,b$, if $\lim_{n\to\infty}C(\Ld_n,s) < +\infty$, then $\omega^{\Ld_n}_s(\cdot) \to \omega^\Gamma_s(\pi^\Gamma_0(\cdot))$ where
	$$\omega^\Gamma_s(\pi^\Gamma_0(\cdot)) = \lb\Omega^\Gamma_s, \pi^\Gamma_0(\cdot) \Omega^\Gamma_s\rb$$
	\item If $\lim_{n\to\infty}C(\Ld_n,s) < +\infty$ for both $s=a,b$, then 
	$\omega^{\Ld_n}_{ab}(\cdot) \to \omega^\Gamma_{ab}(\pi^\Gamma_0(\cdot))$ where
	$$\omega^\Gamma_{ab}(\pi^\Gamma_0(\cdot)) = \lb\Omega^\Gamma_{ab}, \pi^\Gamma_0(\cdot) \Omega^\Gamma_{ab}\rb$$
\end{enumerate}
\end{theorem}

In essence, this theorem states a non-vacuum ground state exists in the infinite volume if the probability of finding the associated particle at any specific site $x$ in the one- or two-particle finite volume ground states, $\ld^{2x}_s/C(\Ld_n,s)$ or $\ld^{2x}_s/C(\Ld_n,ab)$, is strictly positive in the infinite volume limit.  
If the probability goes to zero in the limit, then the infinite volume ground states are indistinguishable from the infinite volume vacuum state. 
 
The proof follows exactly from the structure of the proof of proposition 2.2 in \cite{BHNY2015} for the single species ground states.
The two particle ground state does not add any complications to the proof except for the following convergence argument.
If both $C(\Ld_n,a)$ and $C(\Ld_n,b)$ converge, $C(\Ld_n,ab)$ also converges because it is bounded above by the product $C(\Ld_n,a)C(\Ld_n,b)$.  
If either $C(\Ld_n,a)$ or $C(\Ld_n,b)$ diverge, then $C(\Ld_n,ab)$ also diverges: if   $C(\Ld_n,a)$ diverges (without loss of generality),  then for any $x\in\Ld_1\subseteq\Ld_n$, $C(\Ld_n,ab) > C(\Ld_n\bs\{x\},a)\ld_b^{2x}$ which diverges as well.  

The convergence or divergence of the normalization constants in a given infinite volume $\Gamma$ is determined by the vector $\log\ld_s$.  
If there is a infinite ray $\{t\vec{v}: t\in \mathbb{N}\}$\footnote{or a ray that is thickened by including all points in $\Gamma$ within some fixed finite distance of the ray} contained in $\Gamma$ such that $\log \ld_s \cdot \vec{v} \geq 0$, then the sequence of normalization constants will diverge.  This follows from the fact that the sequence of normalization constants is bounded below by the sum along the rays:
$$\lim_{n\to \infty} C(\Ld_n,s) \geq \lim_{n \to \infty} \sum_{t=1}^{n} e^{t\vec{v}\cdot \log\ld_s} = \infty$$
For any $\log\ld_s$, we can find such an infinite ray in $\mathbb{Z}^d$ and any half-space $\Gamma_{\vec{m}}:= \{x\in\mathbb{Z}^d : \vec{m}\cdot x \geq 0\}$.
Consequently, the PVBS models defined over these spaces have only the vacuum ground state.
Other infinite volumes such as $\mathbb{N}^d \subseteq \mathbb{Z}^d$ may have other infinite volume ground states.
For example, PVBS models defined over $\Gamma=\mathbb{N}^d$ may have other infinite volume ground states if for either $s=a$ or $s=b$, if all $\ld_{s,j} < 1$ for $j=1,\dots,d$ so each $\log\ld_{s,j}<0$.
In such a case,
\begin{align*}
	C(\Gamma,s) &= \lim\limits_{n \to \infty} C([1,n]^d,s)\\
	& = \lim\limits_{n\to\infty} \prod\limits_{j=1}^d \sum_{x_j=1}^n \ld_{s,j}^{2x_j}\\
	& = \lim\limits_{n\to\infty} \prod\limits_{j=1}^d \frac{1-\ld_{s,j}^{2(n+1)}}{1-\ld_{s,j}^{2}}\\
	 & = \prod\limits_{j=1}^d \frac{1}{1-\ld_{s,j}^{2}} 
\end{align*}
and $\omega_s(\cdot)$ is another infinite volume ground state.  
The particle in such a state is localized at the corner of the volume $(1,1,\dots, 1)$ and the probability of finding it elsewhere in $\mathbb{N}^d$ decays exponentially in the distance from that point.  
Note that these ground states are in $\mathcal{H}^\Gamma_0$ and described by applying the quasi-local operators to the vacuum state, see equations \ref{ivqlgs1}, \ref{ivqlgs2}, \ref{ivqlgs3}.  
Essentially, if the $\ld_s^x$ exponentially decays on all rays in the infinite volume, then there is a corresponding single-particle ground state.
If both $\ld_a^x$ and $\ld_b^x$ decay exponentially on all rays, there will also be a two-particle ground state.

\section{Spectral Gap Theorems}

The main result is determining the existence or non-existence of a spectral gap above the ground state for all two-species PVBS models defined on $\mathbb{Z}^d$.  
We define the \textit{spectral gap} of a Hamiltonian $H$ with $H\geq 0$ and $0$ in its spectrum as the quantity
\begin{align}\label{specgap}
	\gamma(H) := \sup\{\delta>0: \text{spectrum}(H) \cap (0,\delta) = \emptyset\}
\end{align}
when it exists and zero when the set on the right side is empty.  
We say an operator is \textit{gapped} if $\gamma(H)>0$ and \textit{gapless} if $\gamma(H)=0$.  
We simplify notation by denoting $\gamma(H^\Ld)$ as $\gamma(\Ld)$ and $\gamma(H^\Gamma)$ as $\gamma(\Gamma)$.  
The existence of a spectral gap is determined by a simple geometric condition on the $\log \ld_s$ vectors.

\begin{theorem}[Spectral Gap on $\mathbb{Z}^d$]\label{specgapthm}
If both $\log\ld_a$ and $\log\ld_b$ are not equal to the zero vector, then $H^{\mathbb{Z}^d}$ is gapped.  If either $\log\ld_a$ or $\log\ld_b$ is equal to the zero vector, then $H^{\mathbb{Z}^d}$ is gapless. 
\end{theorem}

This condition states that if each particle species has a energetically-favored direction, that is, each $\log\ld_s$ is a non-zero vector, then the associated two-species Hamiltonian is gapped.  
If either particle species has no favored direction, then the associated single-species Hamiltonian is gapless and the two-species Hamiltonian is gapless as well.
This result shows the interactions between different particle species do not affect the existence of a spectral gap.
Though the spectral behavior suggests that the system may be embedded into an independent product of two single-species systems, there appears to be no simple embedding.    
In the author's attempts to do so, either the dimension of the ground state space increased which complicates the calculations in section 5
or the terms added in the independent product of two single-species systems to prevent new ground states act non-trivially on the embedded PVBS space. 
The proof for cases where $H^{\mathbb{Z}^d}$ is gapless follows directly from the results in \cite{bishop2016spectral}.  

\begin{proof}[Proof of Gapless Cases:] if either $\log\ld_a = \vec{0}$ or $\log\ld_b = \vec{0}$, the Hamiltonian restricted to the subspace of only $a$-species particles or $b$-species particles, respectively, acts as the single species Hamiltonian on that subspace.  
If $\log\ld_s = \vec{0}$, the single species Hamiltonian is gapless \cite{bishop2016spectral}.  
Therefore, the two species Hamiltonian is gapless as well.\hfill$\Box$\end{proof}

The argument for gapless cases can be directly argued as follows.  
Suppose without loss of generality that $\log\ld_a=\vec{0}$.  
It follows that $\ld_a^x=exp[\log\ld_a \cdot x] = 1$ for all $x$.  
Consider the sequence of finite volumes $\Ld_L:=[0,L-1]^d$ and the finite volume vector $\Psi_a^{\Ld_L}$.  
In the the GNS Hilbert space $\mathcal{H}^{\mathbb{Z}^d}$, this vector is mapped to
\begin{align*}
	\hat{\Psi}_a^{\Ld_L} = \frac{1}{\sqrt{C(\Ld_L,a)}}\sum\limits_{x\in\Ld_L} \lambda_a^x \pi^{\mathbb{Z}^d}_0(p_x^a) \Omega^\Gamma_0
\end{align*}
The only infinite volume ground state is the vacuum state; so this state is orthogonal to the ground state space of  $\mathcal{H}^{\mathbb{Z}^d}$.
The only terms in $H^{\mathbb{Z}^d}$ with positive contributions are projections $h_{x,x+e_j}$ indexed by edges connecting $\Ld_L$ to $\mathbb{Z}^d\bs\Ld_L$.  
(For edges in $\Ld_L$, the vector $\hat{\Psi}_a$ is in the kernel of the corresponding projection as a single particle ground state; for edges in $\mathbb{Z}^d\bs \Ld_L$, the vector $\hat{\Psi}_a$ is in the kernel of the corresponding projection as the vacuum ground state.)  
We denote the set of sites in $\Ld_L$ adjacent to its exterior as $\bd\Ld_L$.  
Each positive contribution is bounded above by the operator norm on projections, 1, times the coefficient squared of the term with a particle at site $x \in \bd\Ld_L$, $\ld_s^{2x}= 1$, times the number of possible edges connect $x$ to the exterior of $\Ld_L$, d.  
The energy (Rayleigh quotient) is bounded above by
\begin{align*}
	\bra{\hat{\Psi}_a^{\Ld_L}}H^{\mathbb{Z}^d}\ket{\hat{\Psi}_a^{\Ld_L}} \leq \frac{d*\sum\limits_{x\in\bd\Ld_L} 1}{\sum\limits_{x\in\Ld_L} 1} = O\left(\frac{1}{L}\right)
\end{align*}
This bounds the spectral gap for all $L$ and thus the PVBS Hamiltonian on $\mathbb{Z}^d$ is gapless.

The proof of a positive spectral gap for $H^{\mathbb{Z}^d}$ is quite difficult and the majority of the technical work in this paper.  
Section \ref{proofofgap} contains the proof for the gapped cases.  Section \ref{prooflemma1} contains the proof of the main technical lemmas needed for the result.  

The results on the spectral gap properties for two-species PVBS models on $\mathbb{Z}^d$ should extend to other infinite subsets $\Gamma$ of $\mathbb{Z}^d$.
For convex volumes bounded by hyperplanes, we need an extra condition: the $\log\ld_s$ vectors cannot be the outward normal of a hyperplane boundary of $\Gamma$.  
Consider half-spaces $\Gamma_{\vec{m}}$ of $\mathbb{Z}^d$, defined by an inward normal vector $\vec{m}\in\mathbb{R}^d$ so $\Gamma_{\vec{m}}:= \{x\in\mathbb{Z}^d : \vec{m}\cdot x \geq 0\}$.
In \cite{bishop2016spectral}, the single-species PVBS model on such half-spaces are gapless if the $\log\ld$ vector pointed in the outward normal direction of $\Gamma_{\vec{m}}$ (that is, a negative scalar multiple of $\vec{m}$) and gapped otherwise.  
The proof of gapless cases follows from the energy for a sequence of finite-volume single particle ground states $\Psi^{\Ld_l}$ pressed up against the boundary.
In the infinite volume, these states are orthogonal to the unique vacuum ground state and has energy of the order $O(L^{-1})$ where $L$ was the linear length of the volume in each direction.  
In Figure \ref{figure:upperbound}, a single particle ground state is supported on the parallelogram $\Ld_L$.  
The terms which contribute positive energy are the edge projections connecting $\Ld_L$ to its exterior; we denote the sites in $\Ld_L$ connected to the exterior by an edge as $\bd\Ld_L$.  
The terms with a particle at site $x$ have magnitude $\exp[2\log\ld x]$ and decay exponentially in their distance from the boundary.  
The gap is bounded above by 
\begin{align*}
\bra{\Psi_s^{\Ld_L}}H^{\Gamma_{\vec{m}}}\ket{\Psi_s^{\Ld_L}}  \leq \frac{d* C(\bd\Ld_L)}{C(\Ld_L)} = O\left(\frac{1}{L}\right)
\end{align*} 
The key to the above bound is that $\ld_a^{2x}$ is maximized across the entire boundary $\Ld_L\bigcap \{\vec{m}\cdot x = 0\}$ which has $O(L^{d-2})$ terms in the numerator and $L^{d-1}$ terms in the denominator.  
This bound holds for all planes of the form $\vec{m}\cdot x =\ constant$  except $\vec{m}\cdot x = L$ where all terms are exponentially small in $L$ and thus neglible.
The bound over each plane exponentially decay and the overall bound follows.   
Therefore, the associated Hamiltonian $H^{\Gamma_{\vec{m}}}$ is gapless.  
For half-spaces, if either $\log\ld_a$ or $\log\ld_b$ point in the outward normal direction of $\Gamma_{\vec{m}}$, the two-species PVBS model is gapless, i.e. $\gamma(H^{\Gamma_{\vec{m}}})=0$.  

\begin{figure}
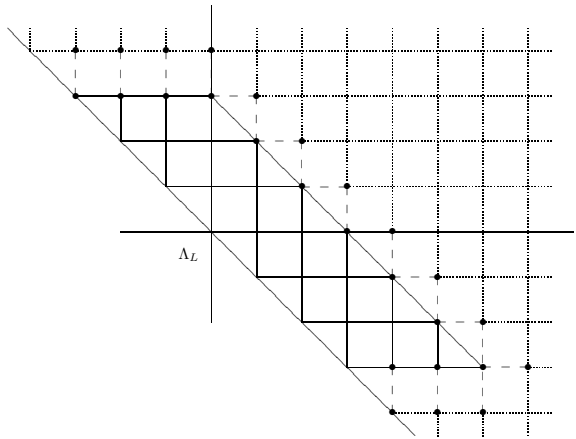


\begin{adjustbox}{max size={.6\textwidth}{.5\textheight}}
{
\xy
(45,-45); (-45,45) **\dir{-};
(-5,-5)*+{\Lambda_L};
(-30,40)*+{\bullet};
(-20,40)*+{\bullet};
(-10,40)*+{\bullet};
(0,40)*+{\bullet};
(-30,30)*+{\bullet};
(-20,30)*+{\bullet};
(-10,30)*+{\bullet};
(0,30)*+{\bullet};
(60,-40)*+{\bullet};
(50,-40)*+{\bullet};
(40,-40)*+{\bullet};
(60,-30)*+{\bullet};
(50,-30)*+{\bullet};
(40,-30)*+{\bullet};
(70,-30)*+{\bullet};
(50,-20)*+{\bullet};
(60,-20)*+{\bullet};
(40,-10)*+{\bullet};
(50,-10)*+{\bullet};
(30,0)*+{\bullet};
(40,0)*+{\bullet};
(20,10)*+{\bullet};
(30,10)*+{\bullet};
(10,20)*+{\bullet};
(20,20)*+{\bullet};
(10,30)*+{\bullet};
(-30,40);(-30,30)**\dir{--};
(-20,40);(-20,30)**\dir{--};
(-10,40);(-10,30)**\dir{--};
(0,40);(0,30)**\dir{--};
(0,30);(10,30)**\dir{--};
(10,20);(10,30)**\dir{--};
(10,20);(20,20)**\dir{--};
(20,10);(20,20)**\dir{--};
(30,10);(30,0)**\dir{--};
(20,10);(30,10)**\dir{--};
(30,0);(40,0)**\dir{--};
(40,0);(40,-10)**\dir{--};
(40,-10);(50,-10)**\dir{--};
(50,-10);(50,-20)**\dir{--};
(50,-20);(60,-20)**\dir{--};
(60,-20);(60,-40)**\dir{--};
(60,-30);(70,-30)**\dir{--};
(50,-30);(50,-40)**\dir{--};
(40,-30);(40,-40)**\dir{--};
(-30,30); (0,30)**\dir{-}; 
(0,30); (60,-30)**\dir{-}; 
(30,-30); (60,-30)**\dir{-};
(-20,20); (10,20)**\dir{-};
(-10,10); (20,10)**\dir{-};
(10,-10); (40,-10)**\dir{-};
(20,-20); (50,-20)**\dir{-};
(-20,30);(-20,20)**\dir{-};
(-10,30);(-10,10)**\dir{-};
(10,20);(10,-10)**\dir{-};
(20,10);(20,-20)**\dir{-};
(30,0);(30,-30)**\dir{-};
(40,-10);(40,-30)**\dir{-};
(50,-20);(50,-30)**\dir{-};
(0,-20);  (0,50) **\dir{-};
(-10,40); (-10,45) **\dir{.};
(-20,40);  (-20,45) **\dir{.};
(-30,40);  (-30,45) **\dir{.};
(-40,40);  (-40,45) **\dir{.};
(10,30);  (10,45) **\dir{.};
(20,20);  (20,45) **\dir{.};
(30,10);  (30,45) **\dir{.};
(40, 0);  (40,45) **\dir{.};
(50,-45);  (50,-40) **\dir{.};
(50,-10);  (50,45) **\dir{.};
(60,-45);  (60,-40) **\dir{.};
(60,-20);  (60,45) **\dir{.};
(70,-45);  (70,45) **\dir{.};
(-20,0);  (80,0) **\dir{-};
(30,10);  (75,10) **\dir{.};
(20,20);  (75,20) **\dir{.};
(10,30);  (75,30) **\dir{.};
(-40,40); (75,40) **\dir{.};
(0,0);  (75,0) **\dir{.};
(50,-10); (75,-10) **\dir{.};
(60,-20); (75,-20) **\dir{.};
(70,-30);(75,-30)**\dir{.};
(40,-40); (75,-40) ** \dir{.};
\endxy 
}
\end{adjustbox}
   \caption{A typical $\Lambda_L$ in $d=2$. An edge between a site on
   the boundary of $\Lambda_L$ and a site outside of $\Lambda_L$
   indicates a nearest neighbor interaction that contributes nonzero energy to
   $\bra{\Psi^L}H^\Gamma\ket{\Psi^L}$. Here $\vec{m}=(1,1)$ and $\log\ld=(-c,-c)$ for some $c>0$.}
   \label{figure:upperbound}
\end{figure}

On the other hand, extending the proofs for gapped cases on $\mathbb{Z}^d$ to other infinite volumes such as half-spaces is much more difficult. 
The key to all these proofs is to find an appropriate sequence of volumes for which the normalization constants have a product (non-trivial) geometric sum structure.
In \cite{bishop2016spectral}, these volumes were enclosed by hyperplanes so the single particle normalization constants have the form
$$ C(\Ld,s) = \prod\limits_{j=1}^d \sum\limits_{x_j=0}^{L_j} \tilde{\ld}_{s,j}^{2x_j}$$
where all $\tilde{\ld}_{s,j} \neq 1$.  
This requires the volumes to satisfy two conditions.  
First, $\tilde{\ld}_s^x$ must be maximized in the corner of the volume and not along an edge or a boundary face of the volume.  
This requires that $\log \tilde{\ld}_s$ essentially point towards a corner of the volume. 
Second, the volume needs to line up with $\mathbb{Z}^d$ to make calculating these products of sums tractable.
The proof for the single-species case required a long list of cases due to the difficulty in finding appropriate volumes which included parallelograms, trapezoids, and their higher dimensional analogues.  
The corresponding proof for the two-species PVBS Hamiltonian is omitted to avoid a longer list of cases due to the second parameter vector. 
These proofs should extend to half-spaces where the  $\log_s\ld$ for $s=a,b$ are non-zero and are not the outward normal to the half-space.
In such cases, the volumes can be constructed in a similar manner to \cite{bishop2016spectral} so the normalization constants for each particle species has a geometric sum structure.
For infinite volumes bounded by a set of hyperplanes, such as $\mathbb{N}^d$, the condition that $\log\ld_s$ should not be pointing in an outward normal should be sufficient to prove a spectral gap.  
It is not clear at this point if these method extends to more exotic volumes bounded by something other than hyperplanes such as a curved surfaces, though the method presented in \cite{kastoryano2017divide} may provide a more robust method in these cases.

\section{Proof of Existence of Spectral Gap for $\mathbb{Z}^d$} \label{proofofgap}
Fix $\log\ld_a \neq \vec{0}$ and $\log\ld_b \neq \vec{0}$.
We prove that $H^{\mathbb{Z}^d}$ is gapped by finding a sequence of bounded, connected, and absorbing finite subvolumes $\Ld_L$ converging to $\mathbb{Z}^d$ (notation: $\Ld_L \nearrow \mathbb{Z}^d$) such that the spectral gap of $H^{\Ld_L}$ is bounded below by a positive constant independent of $L$.
The following theorem concludes this bound is also a lower bound on the spectral gap of the infinite-volume Hamiltonian $H^{\mathbb{Z}^d}$.
\begin{theorem}\label{thm:GNSSpecGap}\cite{fannes1992}
Let $H^\Gamma$ be the GNS Hamiltonian associated with the connected infinite volume
$\Gamma$ with spectral gap $\gamma(\Gamma)$. Then for any sequence of increasing and
absorbing finite volumes $\Lambda_L \nearrow \Gamma$,
\begin{align*}
\gamma(\Gamma) \geq \limsup \gamma(\Lambda_L),
\end{align*}
where $\gamma(\Lambda_L)$ is the spectral gap of the frustration-free
Hamiltonian $H^{\Lambda_L}$.
\end{theorem}

To estimate finite volume gaps, we apply the martingale method. It provides
conditions under which the spectral gap for a frustration-free Hamiltonian on a finite volume $\Lambda$ can be bounded below by a fraction of the spectral gap of sub-volumes. 

\begin{theorem}[Martingale Method, \cite{nachtergaele96}]\label{thm:MartingaleMethod}
For a finite volume $\Lambda$ and frustration-free non-negative definite Hamiltonian $H^\Lambda$, let
$\Lambda_n$ be a finite sequence of volumes with $\Lambda_0 = \emptyset$ and
$\Lambda_n \nearrow\Lambda_L=\Lambda$ such that the following three conditions
hold for the local Hamiltonians for the some $\ell\geq 2$:\\
(i) For some positive constant $d_\ell$,
\begin{equation}\label{mmc1}
\sum_{n=\ell}^L H^{\Lambda_n\backslash\Lambda_{n-\ell}} \leq d_\ell H^{\Lambda_N}
\end{equation}
(ii) For some positive constant $\gamma_\ell$ and $n_\ell$, if $n\geq n_\ell$,
\begin{equation}\label{mmc2}
H^{\Lambda_n\backslash\Lambda_{n-\ell}} \geq \gamma_\ell(\mathbb{I} -
G^{\Lambda_n\backslash\Lambda_{n-\ell}})
\end{equation}
where $G^{\Lambda}$ is the orthogonal projection onto $\mathcal{G}_\Lambda :=
\ker(H^\Lambda)\subset \mathcal{H}^{\Ld}$. \\
(iii) There exists a constant $\epsilon_\ell <
\frac{1}{\sqrt{\ell}}$ and $n_\ell$ such that $n_\ell\leq n\leq N-1$
\begin{equation}\label{mmc3}
\|G^{\Lambda_{n+1}\backslash\Lambda_{n+1-\ell}}E_n\| \leq \epsilon_\ell
\end{equation}
where $E_n := G^{\Lambda_n} - G^{\Lambda_{n+1}}$ is the orthogonal projection onto
$\mathcal{G}_{\Lambda_n}\cap \mathcal{G}^\perp_{\Lambda_{n+1}}$.
Then the spectral gap of $H^{\Lambda_N}$ satisfies
\begin{equation}\label{mmcons}
\gamma(\Lambda_N) \geq \frac{\gamma_\ell}{d_\ell}(1-\epsilon_\ell
\sqrt{\ell})^2.
\end{equation}
\end{theorem}

A sequence of finite volumes which satisfy the three conditions in the martingale method are defined by $d$ pairs of parallel hyperplanes which form the boundaries.  
Each hyperplane is defined in $\mathbb{R}^d$ by the equation $\vec{n}\cdot \vec{x} = constant$.  
To simplify the construction, the normal vectors $\vec{n}$ will be chosen to be two designated vectors $\vec{v}$ and $\vec{w}$ and the standard basis vectors $\vec{e}_j$.  
We choose $\vec{v}$ and $\vec{w}$ so normalization constants for each particle species are products of nontrivial geometric series in terms of the parameters $\ld_{s,j}$.  

As a necessary condition for the operator to be gapped, at least one $\ld_{a,j}$ and one $\ld_{b,k}$ are not equal to one.
It is either the case that this holds for a shared coordinate direction or the case that for each coordinate direction $j =1, \dots, d$, $\ld_{a,j} \neq 1$ and $\ld_{b,j}=1$ or vice-versa.

\textbf{Case 1:} There exists a $j\in\{1,\dots,d\}$ such that $\ld_{a,j}\neq 1$ and $\ld_{b,j}\neq 1$.
We permute coordinate indices so $\ld_{a,1} \neq 1$ and $\ld_{b,1} \neq 1$.
We introduce the vector
\begin{align*}
	\vec{v} = (1,v(2),v(3), \dots, v(d))
\end{align*}
and let $\vec{w}=\vec{e}_2$.  
The other hyperplanes are given by $\vec{x}\cdot\vec{e}_j= x_j = constant$ for $j=3, \dots, d$.  
We choose $v(j) \in \mathbb{N}\cup\{0\}$ for $j=2,\dots,d$ such that
\begin{align*}
	\tilde{\ld}_{a,j} := \ld_{a,j}\ld_{a,1}^{-v(j)} \neq 1\\
	\tilde{\ld}_b(j) := \ld_{b,j}\ld_{b,1}^{-v(j)} \neq 1
\end{align*}
and set $\tilde{\ld}_{s,1} :=\ld_{s,1}$ for $s=a,b$.  
For any $\vec{L}=(L_1,\dots, L_d)\in \mathbb{N}^d$, we define $\Ld_{\vec{L}}$ to be the volume
\begin{align*}
	\{\vec{x}\in\mathbb{Z}^d: 0\leq \vec{v}\cdot \vec{x}\leq L_1-1,\text{ for } j=2,\dots, d,\ 0 \leq  x_j\leq L_j-1\}
\end{align*}
The bounds on $\vec{v}\cdot \vec{x}$ can be expressed as bounds on $x_1$:
\begin{align*}
	-\sum\limits_{j=2}^d v(j)x_j \leq x_1 \leq -\sum\limits_{j=2}^d v(j)x_j + L_1 - 1
\end{align*}
Let $x'_1 = x_1 + \sum\limits_{j=2}^d v(j)x_j$.   
The normalization constants have the following product geometric sum structure for both $s=a,b$:
\begin{align} \label{geometricsum}
	&C(\Ld_{\vec{L}}, s) = \sum\limits_{x\in\Ld_{\vec{L}}} \ld_s^x\nonumber\\
	&=\sum\limits_{x_d=0}^{L_d-1}\ld_{s,d}^{2x_d} \dots \sum\limits_{x_2=0}^{L_2-1}\ld_{s,2}^{2x_2}\sum\limits_{x_1=-\sum_{j=2}^d v(j)x_j}^{L_1-1 -\sum_{j=2}^d v(j)x_j}\ld_{s,1}^{2x_1}\nonumber\\
	&=\sum\limits_{x_d=0}^{L_d-1}(\ld_{s,d}\ld_{s,1}^{-v(d)})^{2x_d} \dots \sum\limits_{x_2=0}^{L_2-1}(\ld_{s,2}\ld_{s,1}^{-v(2)})^{2x_2} \sum\limits_{x'_1=0}^{L_1-1}\ld_{s,1}^{2x'_1}\nonumber\\
	&=\sum\limits_{x_d=0}^{L_d-1}\tilde{\ld}_{s,d}^{2x_d} \dots \sum\limits_{x'_2=0}^{L_2-1}\tilde{\ld}_{s,2}^{2x'_2}\sum\limits_{x'_1=0}^{L_1-1}\tilde{\ld}_{s,1}^{2x_1}
\end{align}
This calculation shows that every sum over points in $\Ld_{\vec{L}}$ with respect to the parameters $\ld_{s,j}$ is equivalent to a sum over points in the box $[0,L_1]\times \dots [0, L_d] \cap \mathbb{Z}^d$ with respect to $\tilde{\ld}_{s,j}$.  
By the choice of volume, each $\tilde{\ld}_{s,j} \neq 1$ so each sum in the product is a nontrivial geometric sum. 

\textbf{Case 2:} There does not exist a $j\in\{1,\dots,d\}$ such that both $\ld_{a,j}\neq 1$ and $\ld_{b,j}\neq 1$.  
By condition of Theorem \ref{specgapthm}, at least one $\ld_{a,j}$ and one $\ld_{b,k}$ are not equal to one.
We permute coordinate indices so $\ld_{a,1}\neq 1$, $\ld_{b,2} \neq 1$, and $\ld_{a,2}=\ld_{b,1}=1$.
We introduce the vectors 
\begin{align*}
	&\vec{v} = \left(\frac{1}{2}, \frac{1}{2}, v(3), v(4), \dots v(d)\right)\\
	&\vec{w} = \left(-\frac{1}{2}, \frac{1}{2}, 0, 0, \dots 0\right)
\end{align*}
We define
\begin{align*} 
	&\tilde{\ld}_{a,1} := \ld_{a,1}\ld_{a,2} = \ld_{a,1} \neq 1\\
	&\tilde{\ld}_{b,1} := \ld_{b,1}\ld_{b,2} = \ld_{b,2} \neq 1\\	
	&\tilde{\ld}_{a,2} := \ld_{a,1}^{-1}\ld_{a,2} = \ld_{a,1}^{-1} \neq 1\\
	&\tilde{\ld}_{b,2} := \ld_{b,1}^{-1}\ld_{b,2} = \ld_{b,2} \neq 1
\end{align*}
and for $j=3,\dots, d$, choose $v(j) \in \mathbb{N}\cup\{0\}$ such that
\begin{align*}
	&\tilde{\ld}_{a,j} := \ld_{a,j}\tilde{\ld}_{a,1}^{-v(j)} \neq 1\\
	&\tilde{\ld}_{b,j} := \ld_{b,j}\tilde{\ld}_{b,1}^{-v(j)} \neq 1
\end{align*}
For any $L=(L_1,\dots, L_d)\in \mathbb{N}^d$, we define $\Ld_{\vec{L}}$ to be 
\begin{align*}
	\{\vec{x}\in\mathbb{Z}^d: 0\leq \vec{v}\cdot \vec{x} \leq L_1-\frac{1}{2},\ 0\leq \vec{w}\cdot\vec{x}\leq L_2-\frac{1}{2},\\
\text{ for } j=3,\dots, d,\ 0 \leq  x_j\leq L_j-1\}
\end{align*}
We partition $\Ld_{\vec{L}}$ into two sets: $E:=\{x\in\Ld_{\vec{L}} : x_1+x_2 \text{ is even}\}$ and $E':=\{x\in\Ld_{\vec{L}}: x_1+x_2 \text{ is odd}\}$.  
By the choice of $\vec{v}$ and $\vec{w}$, $E$ is the set of points where $\vec{v}\cdot\vec{x}$ and $\vec{w}\cdot\vec{x}$ are integer-valued; $E'$ is the set of points where they take half-integer values.  
Moreover, the set $E'$ is $E$ translated by $\vec{e}_2$.
We calculate $C(\Ld_{\vec{L}},s)$ by summing over all points in $E$:
\begin{align*}
	C(\Ld_{\vec{L}},s) &= \sum\limits_{x\in E} \ld_s^{2x} + \sum\limits_{x\in E'}\ld_s^{2x}\\
	&= \sum\limits_{x\in E} \ld_s^{2x} + \sum\limits_{x\in E}\ld_s^{2x+2e_2}\\
	&= (1+\ld_{s,2}^2)\sum\limits_{x\in E} \ld_s^{2x} 
\end{align*}
We introduce coordinates $z_1:= (x_1+x_2)/2$ and $z_2:=(-x_1+x_2)/2$ and note $x_1=z_1-z_2$ and $x_2 = z_1 + z_2$.  
The parameters $\ld_{s,1}$ and $\ld_{s,2}$ associated with the $x_1$ and $x_2$ are equal to the parameters $\tilde{\ld}_{s,1}$ and $\tilde{\ld}_{s,2}$ associated with the $z_1$ and $z_2$:  
\begin{align*}
	\ld_{s,1}^{2x_1}\ld_{s,2}^{2x_2} &= \left(\ld_{s,1}\ld_{s,2}\right)^{2z_1} \left(\ld_{s,1}^{-1}\ld_{s,2}\right)^{2z_2}\\
	&=\tilde{\ld}_{s,1}^{2z_1}\tilde{\ld}_{s,2}^{2z_2}
\end{align*}
The bounds on $\vec{v}\cdot \vec{x}$ and $\vec{w} \cdot \vec{x}$ can be rewritten as bounds on $z_1$ and $z_2$ over $E$:
\begin{align*}
	0\leq \vec{v}\cdot \vec{x} \leq L_1-\frac{1}{2}\\ 
	-\sum\limits_{j=3}^d v(j)x_j \leq \frac{x_1}{2}+\frac{x_2}{2} \leq -\sum\limits_{j=3}^d v(j)x_j + L_1 -  \frac{1}{2}\\
	-\sum\limits_{j=3}^d v(j)x_j \leq z_1 \leq -\sum\limits_{j=3}^d v(j)x_j + L_1 -  \frac{1}{2}\\
	\text{and}\\
	0\leq \vec{w}\cdot\vec{x}\leq L_2-\frac{1}{2}\\
	0\leq \frac{-x_1}{2} + \frac{x_2}{2} \leq L_2 - \frac{1}{2}
	0 \leq z_2 \leq  L_2 - \frac{1}{2}
\end{align*}
where $z_1$ and $z_2$ take integer values in $E$ so the $1/2$ term may be replaced by $1$ in the upper bounds.  
Let $z'_1 := z_1 + \sum_{j=3}^d v(j)x_j$.  
These sum from $0$ to $L_j$ in any sum over $E$.  
The single-species normalization constants have the following product structure for $s=a,b$:
\begin{align*}
	&C(E, s) =\sum\limits_{x\in E} \ld_s^x\\
	&= \sum\limits_{x_d=0}^{L_d-1}\ld_{s,d}^{2x_d} \dots \sum\limits_{x_3=0}^{L_3-1}\ld_{s,3}^{2x_3} \sum\limits_{z_2=0}^{L_2-1}\tilde{\ld}_{s,2}^{2z_2}\sum\limits_{z_1=-\sum_{j=3}^d v(j)x_j}^{L_1-1 -\sum_{j=3}^d v(j)x_j}\tilde{\ld}_{s,1}^{2z_1}\\
	&=\sum\limits_{x_d=0}^{L_d-1}(\ld_{s,d}\tilde{\ld}_{s,1}^{-v(d)})^{2x_d} \dots \sum\limits_{x_3=0}^{L_3-1}(\ld_{s,3}\tilde{\ld}_{s,1}^{-v(3)})^{2x_3} \sum\limits_{z'_2=0}^{L_2-1}\tilde{\ld}_{s,2}^{2z'_2}\sum\limits_{z'_1=0}^{L_1-1}\tilde{\ld}_{s,1}^{2z'_1}\\
	&=\sum\limits_{x_d=0}^{L_d-1}\tilde{\ld}_{s,d}^{2x_d} \dots \sum\limits_{x_3=0}^{L_3-1}\tilde{\ld}_{s,3}^{2x_3} \sum\limits_{z'_2=0}^{L_2-1}\tilde{\ld}_{s,2}^{2z'_2}\sum\limits_{z'_1=0}^{L_1-1}\tilde{\ld}_{s,1}^{2z_1}
\end{align*}
and it follows that 
\begin{align*}
		C(\Ld_{\vec{L}},s) = (1+\ld_{s,2}^2)\sum\limits_{x_d=0}^{L_d-1}\tilde{\ld}_{s,d}^{2x_d} \dots \sum\limits_{x_3=0}^{L_3-1}\tilde{\ld}_{s,3}^{2x_3} \sum\limits_{z'_2=0}^{L_2-1}\tilde{\ld}_{s,2}^{2z'_2}\sum\limits_{z'_1=0}^{L_1-1}\tilde{\ld}_{s,1}^{2z_1}
\end{align*}

For both cases, these volumes are parallelepipeds where the intersection of the volumes with any cross-section parallel to the $(x_1,x_2)$ plane is a rectangle in the first case and diamond-shaped in the second case.  
In both cases, the normalization constants have a product geometric sum structure
\begin{align}
	C(\Ld_{\vec{L}},s) = \kappa\prod\limits_{j=1}^d\sum\limits_{x_j=0}^{L_j-1} \tilde{\ld}_s^{2x_j}
\end{align}
where all $\tilde{\ld}_{s,j} \neq 1$ and $\kappa =1$ or $(1+\ld_{s,2}^2)$ depending on the cases.  
Note $\kappa \geq 1$ in these constructions.

We will apply the martingale method $d$ times, one for each coordinate direction, to
volumes of the form $\Ld_L := \Ld_{(L,L,\dots,L)}$.  
For each $j=1, \dots, d$, let 
$$\Ld^{(j)}_n := \Ld_{(L,\dots, L, n, \ell, \dots, \ell)}$$
where $n$ appears in the $j$-th coordinate.   
We apply the martingale method $d$ times from $j=d$ to $j=1$, each time to the sequence of volumes $\{ \Ld^{(j)}_n \}_{n=\ell}^L$. 
As a necessary condition, $L\geq \ell$, where $\ell$ must be an integer greater than two such that:
\begin{enumerate}
	\item $\ell \geq \max\{v_s(j)+1:\text{ for } s=a,b, \text{ and } j=3, \dots, d\} \label{ellconn}$
	\item $\text{ for both } s=a,b, \text{ and all } j=1, \dots, d,$
	\begin{align}
		\label{ell3} \sqrt{(60\ell)}\tilde{c}^{3/2} \max\left\{ \left(-(\ell-2)|\log\tilde{\ld}_{a,j}|\right), \left(-(\ell-2)|\log\tilde{\ld}_{b,j}|\right)\right\} < \frac{1}{\sqrt{\ell}} 
	\end{align}
\end{enumerate}
where 
\begin{align}
	\tilde{c}:=  \left(1-\prod\limits_{j=1}^d \left(\max\left\{ 1+e^{-2|\log\tilde{\ld}_{a,j}|}, 1+e^{-2|\log\tilde{\ld}_{b,j}|}\right\}\right)^{-1}\right)^{-1} \label{ctilde}
\end{align}
and let
\begin{align}\label{epsl}
	\epsilon_\ell:=  \sqrt{(60\ell)}\tilde{c}^{3/2} \max\left\{ \left(-(\ell-2)|\log\tilde{\ld}_{a,j}|\right), \left(-(\ell-2)|\log\tilde{\ld}_{b,j}|\right)\right\}
\end{align}
which is less than $1/\sqrt{\ell}$ by choice of $\ell$. 
Note $\ell-2 \geq 1/|\log\tilde{\ld}_{s,j}|$ for all $s$ and $j$.

For the $j$-th application of the martingale method, the conditions of the martingale method are satisfied as follows.  
Condition (i) of the martingale method is satisfied by $d_\ell = \ell$ by the following argument.  
Each edge $x,x+e_k$ appears in at most $\ell$ different $\Ld^{(j)}_n\bs\Ld^{(j)}_{n-\ell}$ since these volumes are translations of one another in the $j$-th coordinate direction.  
The corresponding edge projections $h_{x,x+e_j}$ in the inequality \ref{mmc1} appear at most $\ell$ times on the left side of inequality and $d_\ell$ times on the right side.  Since all the terms are projections and thus non-negative definite, the inequality 
$$\sum_{n=\ell}^L H^{\Lambda^{(j)}_n\backslash\Lambda^{(j)}_{n-\ell}} \leq d_\ell H^{\Lambda^{(j)}_N}$$ 
is satisfied.  

Condition (ii) of the martingale method is satisfied for $\gamma' = \gamma(\Ld^{(j)}_\ell)$.  
Each $\Ld^{(j)}_n\bs\Ld^{(j)}_{n-\ell}$ is a translation of $\Ld^{(j)}_\ell$.
The PVBS Hamiltonians are translation-invariant, so each $H^{\Ld^{(j)}_n\bs\Ld^{(j)}_{n-\ell}}$ is unitarily equivalent to $H^{\Ld^{(j)}_\ell}$ and thus have the same spectral gap.  
This is a finite dimensional operator and therefore has a positive lower bound on the spectral gap.  

Condition (iii) is satisfied for the the choice of $\ell$ and $\epsilon_\ell$ which satisfy the definition above.
The first condition on $\ell$ ensures the volumes $\Ld^{(j)}_{n+1}\bs\Ld^{(j)}_{n+1-\ell}$ are connected.  
This follows from inspection of the projections into the two-dimensional coordinate planes.  
The second condition on $\ell$ ensures that the bound in the following lemma satisfies conditions of the martingale method.

\begin{lemma}[Two-species Projection Bound]\label{mmc3lemma}
Let $G^{\Ld}$ be the projection onto the ground state space of $\mathcal{H}^\Ld$, denoted $\mathcal{G}^\Ld$.  
Let $E^{(j)}_n := G^{\Lambda^{(j)}_n} - G^{\Lambda^{(j)}_{n+1}}$ which is the projection onto
$\mathcal{G}_{\Lambda^{(j)}_n}\cap \mathcal{G}^\perp_{\Lambda^{(j)}_{n+1}}$.  
Suppose $\tilde{\ld}_{s,j}$ is not equal to one for $s =a,\ b$ and for all $j=1,\dots, d$.  
Suppose that $C(\Ld^{(j)}_n,s)$ is a product of geometric series of $\tilde{\ld}_s(j)$.  
Suppose $\ell$ is an integer large enough so $\Ld^{(j)}_n \bs \Ld^{(j)}_{n-\ell}$ is connected for all $n\geq \ell$ and $\ell -2$ is greater than both $(\log\tilde{\ld}_{a,j})^{-1}$ and $(\log\tilde{\ld}_{a,j})^{-1}$.
For $n\geq \ell$, we have
\begin{flalign}
	&\| G^{\Ld^{(j)}_{n+1}\bs\Ld^{(j)}_{n+1-\ell}}E^{(j)}_n\| \nonumber \\
	&\leq  \sqrt{(60\ell)}\tilde{c}^{3/2} \max\left\{ \left(-(\ell-2)|\log\tilde{\ld}_{a,j}|\right), \left(-(\ell-2)|\log\tilde{\ld}_{b,j}|\right)\right\}
\end{flalign}
\end{lemma}
The proof of this lemma is section 5.  
By choice of $\ell$, we have
\begin{align}
	\|G^{\Ld^{(j)}_{n+1}\bs\Ld^{(j)}_{n+1-\ell}}E^{(j)}_n\| \leq \epsilon_\ell < \frac{1}{\sqrt{\ell}}
\end{align}
and condition (iii) of the martingale method is satisfied.
With the conditions of the martingale method satisfied, the spectral gap $\gamma(\Ld^{(j)}_{L})$ is bounded below by 
\begin{align*}
	\gamma(\Ld^{(j)}_{L}) \geq \gamma(\Ld^{(j)}_\ell) \frac{\left(1-\epsilon_\ell\sqrt{\ell}\right)^2}{\ell}
\end{align*}

By the definition of the volumes, we have $\Ld_L = \Ld^{(d)}_L$ and $\Ld^{(j)}_\ell = \Ld^{(j-1)}_L$.  
Substituting the spectral gap bounds for each coordinate directions, we have
\begin{align*}
	\gamma(\Ld_L) &=\gamma(\Ld^{(d)}_L)\geq \gamma(\Ld^{(d)}_\ell) \frac{\left(1-\epsilon_\ell\sqrt{\ell}\right)^2}{\ell}\\
	&=\gamma(\Ld^{(d-1)}_L) \frac{\left(1-\epsilon_\ell\sqrt{\ell}\right)^2}{\ell}\geq\gamma(\Ld^{(d-1)}_\ell) \left(\frac{\left(1-\epsilon_\ell\sqrt{\ell}\right)^2}{\ell}\right)^2\\
	&\dots \\
	&=\gamma(\Ld^{(j)}_L) \left(\frac{\left(1-\epsilon_\ell\sqrt{\ell}\right)^2}{\ell}\right)^{(d-j)}\geq\gamma(\Ld^{(j)}_\ell) \left(\frac{\left(1-\epsilon_\ell\sqrt{\ell}\right)^2}{\ell}\right)^{(d-j+1)}\\
	&=\gamma(\Ld^{j-1}_L)\left(\frac{\left(1-\epsilon_\ell\sqrt{\ell}\right)^2}{\ell}\right)^{(d-j+1)}\geq\gamma(\Ld^{(j-1)}_\ell) \left(\frac{\left(1-\epsilon_\ell\sqrt{\ell}\right)^2}{\ell}\right)^{(d-j+2)}\\
	&\dots \\
	&=\gamma(\Ld_L^{(1)})\left(\frac{(1-\epsilon_\ell\sqrt{\ell})^2}{\ell}\right)^{d-1} \geq\gamma(\Ld^{(1)}_\ell) \left(\frac{\left(1-\epsilon_\ell\sqrt{\ell}\right)^2}{\ell}\right)^d
\end{align*}
This last bound is strictly positive and independent of $L$.  
By translational invariance, this bound also holds for the PVBS Hamiltonian defined over the volume
$$\{\vec{x}\in\mathbb{Z}^d: -L\leq \vec{v}\cdot \vec{x}\leq L, -L \leq \vec{w}\cdot \vec{x} \leq L,\text{ for } j=3,\dots, d,\ -L \leq x_j\leq L\}$$
This sequence converges to $\mathbb{Z}^d$ as $L\to\infty$.
By the Theorem \ref{thm:GNSSpecGap}, this positive lower bound also bounds the spectral gap of $H^{\mathbb{Z}^d}$.
Therefore, $\gamma(\mathbb{Z}^d) > 0$ and $H^{\mathbb{Z}^d}$ is gapped. \hfill $\Box$


\section{Ground State Space Proofs}\label{groundstates}
\begin{proof}[Structure of Finite Volume Ground State Space]:
The proof follows the structure of the proof of proposition 2.1 in \cite{BHNY2015}.  
 
Let $\Ld$ be a bounded, connected, finite subset of $\mathbb{Z}^d$.   
The ground state space is the kernel of $H^\Ld$ because it is the sum of non-negative edge projections.  
A vector is in the kernel of $H^\Ld$ if and only if it is in the kernel of $h_{x,x+e_j}$ for every edge $x,x+e_j$ in $\Ld$.  
The ground state space can be decomposed into subspaces where the number $a$ and $b$ species particles are fixed since the Hamiltonian preserves those numbers. 

\textbf{(i)} The subspace with no particles is spanned by the finite volume vacuum state $\Psi_0^\Ld$.  Each edge operator $h_{x,x+e_j}$ sends this vector to zero.  Therefore, it is in the ground state space.

\textbf{(ii)} The subspace with exactly one particle of one species and no particle of the other species is unitarily equivalent to the single-species ground state; the proof from \cite{BHNY2015} is repeated for presentation.  

Consider the subspace with only one particle of either species and none of the other.
Without loss of generality, suppose it is species $a$.  
This subspace of the Hilbert space is spanned by vectors of the form 
\begin{align*}
	\Psi = \sum\limits_{x\in\Ld} c_x \ket{a}_x
\end{align*}
For each edge $x, x+e_j$ in $\Ld$, $\Psi$ is in the kernel of the edge projection if and only if the two terms $c_x\ket{a}_x + c_{x+e_j}\ket{a}_{x+e_j}$, i.e. $c_x\ket{a,0} + c_{x+e_j}\ket{0,a}$, is a constant multiple of 
$\ket{a,0} + \ld_{a,j}\ket{0,a}$.  
This holds if and only if
\begin{align}\label{stepa}
	 c_{x+e_j}= \ld_{a,j} c_x 
\end{align}
Suppose $c_x$ is determined at a specified $x$.  
For any other $x' \in \Ld$, there exists a path of adjacent vertices $p_0, p_1, \dots, p_{T-1}, p_T$ that connects $x$ to $x'$ in $\Ld$ with $x=p_0$ and $x'=p_T$.  
Equation \eqref{stepa} holds for each edge connecting adjacent vertices $p_{k-1},p_k$:
\begin{align*}
	c_{p_k} = \ld_a^{p_k-p_{k-1}}c_{p_{k-1}}
\end{align*}
where $p_k-p_{k-1} =\pm e_j$ for some coordinate direction $j$.  
Combining these equations along the path from $x$ to $x'$, the value of $c_{x'}$ is determined by $c_x$:
\begin{align*}
	c_{x'}= \ld_a^{x'-x} c_x .
\end{align*}
Therefore, the ground state space in the $a$-particle subspace is one dimensional.  
The normalized ground state vector in the $a$-particle subspace is $\Psi_a^\Ld$.  
By interchanging $b$ for $a$ in the argument above, it follows that $\Psi_b^\Ld$ is the normalized ground state vector in the $b$-particle subspace.

\textbf{(iii)} The subspace with exactly one particle of each species is spanned by vectors of the form
\begin{align*}
	\Psi = \sum\limits_{\substack{x,y\in\Lambda\\ y\neq x}} c_{x,y} \ket{a}_x\otimes\ket{b}_y
\end{align*}
For each edge in $\Ld$, $\Psi$ is in the kernel of the edge projection if and only if it is a constant multiple of 
$$\ket{0,0},\ 
\ket{a,0} + \ld_{a,j}\ket{0,a},\  
\ket{b,0} + \ld_{b,j}\ket{0,b} ,\  
\ld_{b,j}\ket{a,b} + \ld_{a,j}\ket{b,a}$$
which is equivalent to $c_{x,y}$ satisfying
\begin{align}
	&c_{x+e_j,y} = \ld_{a,j}c_{x,y} & \text{ if } y\neq x \text{ or } x+e_j \label{ahop}\\
	&c_{x,y+e_j} = \ld_{b,j}c_{x,y} & \text{ if } x\neq y \text{ or } y+e_j\label{bhop}\\
	&c_{x,x+e_j} = \frac{\ld_{b,j}}{\ld_{a,j}} c_{x+e_j,x} & \text{ if $x$ and $y$ are adjacent.}\label{abhop} 
\end{align}
The choice $c_{x,y} = \ld_a^x\ld_b^y$ satisfies these conditions.  

All other solutions are a constant multiple of this solution by the following argument.  
Suppose $x,y,x',y'$ are points in $\Ld$ such that $x\neq y$ and $x' \neq y'$.
The volume is connected so there exists a path from $x$ to $x'$ in $\Ld$, $p_0, p_1, \dots, p_T$ where $p_0=x$ and $p_T=x'$.  
For each edge along the path, equations \eqref{ahop} or \eqref{abhop} hold with the latter occurring if $y$ is on the path.
If $y$ is on the path, then equation \eqref{abhop} is satisfied for the edge that first has $y$ as an end, $y=p_{k'}$.  
In this case, let $y'':=p_{k'-1}$, the vertex on the other end.  
If $y$ is not on the path, we let $y'':=y$.
We have 
$$c_{x',y''} = \ld_a^{x'-x}\ld_b^{y''-y}c_{x,y}$$
There also exists a path from $y''$ to $y$.  
For each edge along the path, equations \eqref{bhop} or \eqref{abhop} hold.
If $x'$ is on the path, then equation \eqref{abhop} is satisfied for the edge that first has $x'$ as an end, let $x''$ denote the other end.  
If not, let $x''=x'$.
We have
$$c_{x'',y'} = \ld_a^{x''-x'}\ld_b^{y'-y''}c_{x',y''}$$
To satisfy equation \eqref{ahop} for the edge joining $x'$ and $x''$, 
$$c_{x',y'} = \ld_a^{x'-x''}c_{x'',y'}$$
Combining these equations, we have
$$c_{x',y'} = \ld_a^{x'-x}\ld_b^{y'-y} c_{x,y}$$
and thus all $c_{x',y'}$ are determined by the value of $c_{x,y}$.
Therefore, there is the ground state space restricted to this is subspace is one dimensional and spanned by $\Psi_{ab}^\Ld$.

\textbf{(iv)} The particle space with two or more particles of a specific species is spanned by vectors of the form:
$$ \Psi = \sum\limits_{A,B \subset \Ld, \substack{A\cap B = \emptyset}} c_{A,B}\bigotimes\limits_{x\in A}\ket{a}_x \otimes \bigotimes\limits_{y\in B} \ket{b}_y $$
where $A$ is the set of sites occupied by a particle of species $a$ and $B$ is the set of sites occupied by a particle of species $b$.  
The only vector $\Psi$ in this subspace which are also in the ground state space is the zero vector, that is, $c_{A,B}$ is zero whenever $|A|\geq 2$ or $|B|\geq 2$.

Suppose that $|A|\geq 2$ or $|B|\geq 2$ and fix them.  
For each set with more than one element, there exist a pair of distinct points closest to each other in graph distance.  
We choose the pair of points closest to each other; without loss of generality, assume these points are $x$ and $x'$ in $A$.  
If $x$ and $x'$ are connected by an edge, then the term of $\Psi$ corresponding to $A,B$ can be rewritten as 
\begin{align}
	c_{A,B} \ket{a,a}_{x,x'} \otimes\bigotimes\limits_{x''\in A\bs\{x,x'\}}\ket{a}_{x''} \otimes \bigotimes\limits_{y\in B} \ket{b}_y  
\end{align} 
The $x,x'$ edge projection acts as identity on this term:
\begin{align*}
	h_{x,x'}&c_{A,B} \ket{a,a}_{x,x'} \otimes\bigotimes\limits_{x''\in A\bs\{x,x'\}}\ket{a}_{x''} \otimes \bigotimes\limits_{y\in B} \ket{b}_y  \\
	&= \ket{a,a}\bra{a,a}_{x,x'} c_{A,B} \ket{a,a}_{x,x'} \otimes\bigotimes\limits_{x''\in A\bs\{x,x'\}}\ket{a}_{x''} \otimes \bigotimes\limits_{y\in B} \ket{b}_y \\
	&=  c_{A,B} \ket{a,a}_{x,x'} \otimes\bigotimes\limits_{x''\in A\bs\{x,x'\}}\ket{a}_{x''} \otimes \bigotimes\limits_{y\in B} \ket{b}_y
\end{align*} 
The vector $\Psi$ is in the ground state space if and only if this term is zero, that is, $c_{A,B} = 0$.

If $x$ and $x'$ are not adjacent,  there exists a path from $x$ to $x'$, $p_0, p_1, p_2, \dots, p_T$.  
There is at most one $y$ in $B$ along this path by our assumption that $x$ and $x'$ are the closest pair of points which are occupied by the same particle species.  
We define a finite sequence of sets indexed from $0$ to $T-1$. 
Let $A(0) :=A$, $B(0):=B$, and for $j=1, \dots, T-1$, 
$$A(k) := A(k-1) \cup \{p_k\}\bs\{p_{k-1}\}$$
\[ B(k) := \begin{cases} 
      B(k-1) & \text{ if } y \neq p_k \\
      B(k-1) \cup \{p_{k-1}\} \bs\{p_k\} & \text{ if } y=p_k \end{cases}
\]
These sets track the locations of $a$ and $b$ particles as we apply the ground state conditions associated with the edges in the path.  
In essence, each edge projection effectively moves the $a$ particle along the path and exchanges it with the $b$ particle when it encounters it.  
If $y$ is on the path at $p_{k'}$, let $y'':= p_{k'-1}$; otherwise, let $y'':= y$.
For $\Psi$ to be in the ground state space, the $c_{A(k),B(k)}$ must satisfy equations \eqref{ahop}, \eqref{bhop}, and \eqref{abhop}.  
For the sequence of sets above, this requires
\[ c_{A(k),B(k)} = \begin{cases} 
      \ld_a^{p_k-p_{k-1}} c_{A(k-1),B(k-1)} & \text{ if } y \neq p_k \\
      \ld_a^{p_k-p_{k-1}} \ld_b^{p_{k-1}- p_k} c_{A(k-1),B(k-1)} & \text{ if } y=p_k \end{cases}
\]
and combining these equations gives  
$$ c_{A(T-1), B(T-1)} = \ld_a^{p_{T-1}-x}\ld_b^{y''-y} c_{A,B}$$
The set $A(T-1)$ contains adjacent points $p_{T-1}$ and $x'$.
For $\Psi$ to be in the ground state space, this $c_{A(T-1),B(T-1)} = 0$ as before when $x$ and $x'$ were adjacent.  
It follows that $c_{A,B} = 0$.
The only solution for a vector with two or more particles of a single species is the zero vector.  
Therefore, there is no nonzero ground state vector in the subspace with two or more particles of a single species.
\end{proof}

\section{Proof of Lemma 1 }\label{prooflemma1}

This section proves Lemma \eqref{mmc3lemma} which is necessary to prove condition (3) of the martingale method.  
Restated, when $\tilde{\ld}_s(j)$ is not equal to one for $s =a,\ b$ and $C(\Ld^{(j)}_n,s)$ is a product of geometric series of $\tilde{\ld}_s(j)$ and $\ell-2 \geq \max\{(\log\tilde{\ld}_{s,j})^{-1}: s=a,b,\ j=1,\dots,d \}$, then
\begin{flalign}
	&\| G^{\Ld^{(j)}_{n+1}\bs\Ld^{(j)}_{n+1-\ell}}E^{(j)}_n\| \nonumber \\
	&\leq  \sqrt{(60\ell)}\tilde{c}^{3/2} \max\left\{ \left(-(\ell-2)|\log\tilde{\ld}_{a,j}|\right), \left(-(\ell-2)|\log\tilde{\ld}_{b,j}|\right)\right\}
\end{flalign}
where $G^\Ld$ is the projection onto the ground state space in $\mathcal{H}^\Ld$ and $E_n^{(j)}:= G^{\Ld^{(j)}_{n}} -G^{\Ld^{(j)}_{n+1}}$.  

This section is divided up as follows.
In subsection \ref{NCRB}, we prove ratios of normalization constants and diagonal terms are bounded or (nearly) exponentially small in $\ell$ and $n$.
In subsection \ref{projonb}, we find an operator norm bound on $\|G^{\Ld_{n+1}\bs\Ld_{n+1-\ell}}E_n\|^2$ for the two-species PVBS model for general sequences of volumes $\{ \Ld\}_{n=0}^{N}$.  For each subspace, we apply the bounds in \ref{NCRB} to bound the operator norms on subspaces.
The maximum of these operator norms obtains the bound above. 

\subsection{Normalization Constant Ratio Bounds} \label{NCRB}

An intermediate step to proving the bounds in Lemma \eqref{mmc3lemma} is proving that ratios of normalization constants over various subsets of $\Ld$ have either finite bounds or exponentially small bounds in $\ell$ or $n$.

By construction, the $C(\Ld^{(j)}_{n},s)$ have the form 
\begin{align}
	C(\Ld^{(j)}_{n},s) = \kappa\left(\prod\limits_{k=1}^{j-1}\sum\limits_{x_k=0}^{L-1} \tilde{\ld}_{s,k}^{2x_k}\right)\left(\sum\limits_{x_j=0}^{n-1} \tilde{\ld}_{s,j}^{2x_j}\right)\left(\prod\limits_{k=j+1}^d\sum\limits_{x_k=0}^{\ell-1} \tilde{\ld}_{s,k}^{2x_k}\right)
\end{align}
where $s =a \text{ or } b$, all $\tilde{\ld}_{s,j}$ are positive and not equal to one, and $\kappa =1$ or $\kappa = (1+\ld_s(2)^2) \geq 1$, depending on the volume constructed.  
The $\kappa$ does not meaningfully affect the bounds.
The diagonal term $D(\Ld) := C(\Ld,a)C(\Ld,b) - C(\Ld,ab)$ has a similar structure:
\begin{align}
	D(\Ld^{(j)}_{n}) = \kappa\left(\prod\limits_{k=1}^{j-1}\sum\limits_{x_k=0}^{L-1} \left(\tilde{\ld}_{a,k}\tilde{\ld}_{b,k}\right)^{2x_k}\right)\left(\sum\limits_{x_j=0}^{n-1} \left(\tilde{\ld}_{a,j}\tilde{\ld}_{b,j}\right)^{2x_j}\right)\nonumber \\ \cdot\left(\prod\limits_{k=j+1}^d\sum\limits_{x_k=0}^{\ell-1} \left(\tilde{\ld}_{a,k}\tilde{\ld}_{b,k}\right)^{2x_k}\right)
\end{align}
This structure also applies to $\Ld^{(j)}_n\bs\Ld^{(j)}_m$:
\begin{align*}
	C(\Ld^{(j)}_n\bs\Ld^{(j)}_m,s) = \kappa\left(\prod\limits_{k=1}^{j-1}\sum\limits_{x_k=0}^{L-1} \tilde{\ld}_{s,k}^{2x_k}\right)\left(\sum\limits_{x_j=m}^{n-1} \tilde{\ld}_{s,j}^{2x_j}\right)\left(\prod\limits_{k=j+1}^d\sum\limits_{x_k=0}^{\ell-1} \tilde{\ld}_{s,k}^{2x_k}\right)\\
		D(\Ld^{(j)}_n\bs\Ld^{(j)}_m) = \kappa\left(\prod\limits_{k=1}^{j-1}\sum\limits_{x_k=0}^{L-1} \left(\tilde{\ld}_{a,k}\tilde{\ld}_{b,k}\right)^{2x_k}\right)\left(\sum\limits_{x_j=m}^{n-1} \left(\tilde{\ld}_{a,j}\tilde{\ld}_{b,j}\right)^{2x_j}\right)\nonumber \\ \cdot\left(\prod\limits_{k=j+1}^d\sum\limits_{x_k=0}^{\ell-1} \left(\tilde{\ld}_{a,k}\tilde{\ld}_{b,k}\right)^{2x_k}\right)
\end{align*}
Any result for $C(\Ld^{(j)}_n\bs\Ld^{(j)}_m,s)$ applies for $C(\Ld^{(j)}_n,s)$ because $\Ld^{(j)}_{0} = \emptyset$ and thus $\Ld^{(j)}_n = \Ld^{(j)}_n\bs\Ld^{(j)}_{0}$.

The bounds that are exponential in $\ell$ can be expressed in a single form for both when $\ld_{s,j}$ is greater than one and when it is less than one, we note 
if $\ld >1$, then $\log \ld = |\log \ld|$ and it follows $$\ld^{-\ell} = \exp(-\ell\log\ld) = \exp(-\ell|\log\ld|);$$
and if $\ld <1$, then $\log\ld = -|\log\ld|$ and $$\ld^{\ell} = \exp(\ell\log\ld) = \exp(-\ell|\log\ld|).$$
In both cases, the expression+ decays exponentially in $\ell$.


\begin{lemma}[Product Bounds]\label{productbounds}
Suppose $\Ld^{(j)}_n,\ \Ld^{(j)}_n\bs\Ld^{(j)}_m$ and $C(\Ld^{(j)}_n,s),\ C(\Ld^{(j)}_n\bs\Ld^{(j)}_m,s)$ have the product geometric sum structure described above.  
Suppose $L \geq \ell\geq 2, n\geq 2$ or $n-m\geq2$.  
If all $\tilde{\ld}_{s,k} \neq 1$ for all $s=a,b$ and $k=1,\dots,d$, then 
\begin{align} 
	C(\Ld^{(j)}_n\bs\Ld^{(j)}_m,ab) \leq C(\Ld^{(j)}_n\bs\Ld^{(j)}_m,a)C(\Ld^{(j)}_n\bs\Ld^{(j)}_m,b)\\
	C(\Ld^{(j)}_n\bs\Ld^{(j)}_m,a)C(\Ld^{(j)}_n\bs\Ld^{(j)}_m,b) \leq \tilde{c} C(\Ld^{(j)}_n\bs\Ld^{(j)}_m,ab) 
\end{align}
where
\begin{align}
	\tilde{c}:=  \left(1-\prod\limits_{k=1}^d \left(\max\left\{ 1+e^{-2|\log\tilde{\ld}_{a,k}|}, 1+e^{-2|\log\tilde{\ld}_{b,k}|}\right\}\right)^{-1}\right)^{-1}
\end{align}
which is greater than 1. 
\end{lemma}

\textbf{Remark:} These bounds will applied to split $C(\Ld,ab)$ in the following forms:
\begin{align}
	C(\Ld,ab) \leq C(\Ld,a)C(\Ld,b)\\
	\frac{1}{C(\Ld,ab)} \leq \frac{\tilde{c}}{C(\Ld,a)C(\Ld,b)}
\end{align}
the former for $C(\Ld,ab)$ which appear in numerators and latter for $C(\Ld,ab)$ which appear in the denominator.

\begin{proof}
The first bound follows from the fact that in $C(\Ld,ab) = C(\Ld,a)C(\Ld,b) - D(\Ld)$ and $D(\Ld)\geq 0$.\\

For the second bound, we decompose into a product of ratios of sums:
\begin{align*}
\frac{D(\Ld^{(j)}_n\bs\Ld^{(j)}_m)}{C(\Ld^{(j)}_n\bs\Ld^{(j)}_m,a)C(\Ld^{(j)}_n\bs\Ld^{(j)}_m,b)} \leq \left(\prod\limits_{k=1}^{j-1}\frac{\sum\limits_{x_k=0}^{L-1} \left(\tilde{\ld}_{a,k}\tilde{\ld}_{b,k}\right)^{2x_k}}{ \sum\limits_{x_k=0}^{L-1} \tilde{\ld}^{2x_k}_{a,k} \sum\limits_{x_k=0}^{L-1}\tilde{\ld}^{2x_k}_{b,k}}\right)\\
	 \cdot\left(\frac{\sum\limits_{x_j=m}^{n-1} \left(\tilde{\ld}_{a,j}\tilde{\ld}_{b,j}\right)^{2x_j}}{\sum\limits_{x_j=m}^{n-1} \tilde{\ld}^{2x_j}_{a,j}\sum\limits_{x_j=m}^{n-1}\tilde{\ld}^{2x_j}_{b,j}}\right)
	\cdot\left(\prod\limits_{k=j+1}^d\frac{\sum\limits_{x_k=0}^{\ell-1} \left(\tilde{\ld}_{a,k}\tilde{\ld}_{b,k}\right)^{2x_k}}{\sum\limits_{x_k=0}^{\ell-1} \tilde{\ld}^{2x_k}_{a,k}\sum\limits_{x_k=0}^{\ell-1}\tilde{\ld}^{2x_k}_{b,k}}\right)\\
	=\left(\prod\limits_{k=1}^{j-1}\frac{\sum\limits_{x_k=0}^{L-1} \left(\tilde{\ld}_{a,k}\tilde{\ld}_{b,k}\right)^{2x_k}}{ \sum\limits_{x_k=0}^{L-1} \tilde{\ld}^{2x_k}_{a,k} \sum\limits_{x_k=0}^{L-1}\tilde{\ld}^{2x_k}_{b,k}}\right)\\
	 \cdot\left(\frac{\sum\limits_{x_j=0}^{n-m-1} \left(\tilde{\ld}_{a,j}\tilde{\ld}_{b,j}\right)^{2x_j}}{\sum\limits_{x_j=0}^{n-m-1} \tilde{\ld}^{2x_j}_{a,j}\sum\limits_{x_j=0}^{n-m-1}\tilde{\ld}^{2x_j}_{b,j}}\right)
	\cdot\left(\prod\limits_{k=j+1}^d\frac{\sum\limits_{x_k=0}^{\ell-1} \left(\tilde{\ld}_{a,k}\tilde{\ld}_{b,k}\right)^{2x_k}}{\sum\limits_{x_k=0}^{\ell-1} \tilde{\ld}^{2x_k}_{a,k}\sum\limits_{x_k=0}^{\ell-1}\tilde{\ld}^{2x_k}_{b,k}}\right)\\
\end{align*} 
where $L_j = n-m$ and the extra $\kappa$ in the denominator was dropped since it is greater than 1.
We bound each of the ratios of geometric sums.
We will use $L_k$ for $L,\ n, \text{ or }, \ell$.  
In all cases, $L_k \geq \ell \geq 2$.  \\
\textit{Case 1: Suppose $\tilde{\ld}_{a,k} >1$.}  Then 
\begin{align*}
	\frac{\sum\limits_{x_k=0}^{L_k-1} \tilde{\ld}_{a,k}^{2x_k}\tilde{\ld}_{b,k}^{2x_k}}{\left(\sum\limits_{x_k=0}^{L_k-1} \tilde{\ld}_{a,k}^{2x_k}\right)\left(\sum\limits_{x_k=0}^{L_k-1} \tilde{\ld}_{b,k}^{2x_k}\right)} &\leq \frac{\tilde{\ld}_{a,k}^{2(L_k-1)} \sum\limits_{x_k=0}^{L_k-1} \tilde{\ld}_{b,k}^{2x_k}(k)}{\left(\sum\limits_{x_k=0}^{L_k-1} \tilde{\ld}_{a,k}^{2x_k}\right)\left(\sum\limits_{x_k=0}^{L_k-1} \tilde{\ld}_{b,k}^{2x_k}\right)}\\
	&= \frac{\tilde{\ld}_{a,k}^{2(L_k-1)}\left(\tilde{\ld}_{a,k}^2-1\right)}{\tilde{\ld}_{a,k}^{2L_k} - 1}\\
	&= \frac{\tilde{\ld}_{a,k}^2-1}{\tilde{\ld}_{a,k}^2 - \tilde{\ld}_{a,k}^{-2(L_k-1)}}\\
	&\leq \frac{\tilde{\ld}_{a,k}^2-1}{\tilde{\ld}_{a,k}^2 - \tilde{\ld}_{a,k}^{-2}}\\
	&= \frac{1 - \tilde{\ld}_{a,k}^{-2}}{1 - \tilde{\ld}_{a,k}^{-4}} \\
	&= \frac{1}{1 + \tilde{\ld}_{a,k}^{-2}}\\
	&= \left(1 + e^{-2|\log\tilde{\ld}_{a,k}|}\right)^{-1}\\
	&< 1
\end{align*}

\textit{Case 2: Suppose $\tilde{\ld}_{a,k} < 1$.} Then
\begin{align*}
	\frac{\sum\limits_{x_k=0}^{L_k-1} \tilde{\ld}_{a,k}^{2x_k}\tilde{\ld}_{b,k}^{2x_k}}{\left(\sum\limits_{x_k=0}^{L_k-1} \tilde{\ld}_{a,k}^{2x_k}\right)\left(\sum\limits_{x_k=0}^{L_k-1} \tilde{\ld}_{b,k}^{2x_k}\right)} &\leq \frac{\sum\limits_{x_k=0}^{L_k-1} \tilde{\ld}_{b,k}^{2x_k}}{\left(\sum\limits_{x_k=0}^{L_k-1} \tilde{\ld}_{a,k}^{2x_k}\right)\left(\sum\limits_{x_k=0}^{L_k-1} \tilde{\ld}_{b,k}^{2x_k}\right)}\\
	&= \frac{1-\tilde{\ld}_{a,k}^2}{1 - \tilde{\ld}_{a,k}^{2L_k}}\\
	&\leq \frac{1-\tilde{\ld}_{a,k}^2}{1 - \tilde{\ld}_{a,k}^{4}}\\
	&= \frac{1}{1 + \tilde{\ld}_{a,k}^{2}} \\
	&= \left(1 + e^{-2|\log\tilde{\ld}_{a,k}|}\right)^{-1}\\
	&<1
\end{align*}
The cases applied to $\tilde{\ld}_{b,k}$ generate the same bounds with $b$ and $a$ swapped.  We may bound each of these by the minimum over $a$ and $b$.  The overall bound is the product of the coordinate wise bounds.
\begin{align*}
	\frac{D(\Ld^{(j)}_n\bs\Ld^{(j)}_m)}{C(\Ld^{(j)}_n\bs\Ld^{(j)}_m,a)C(\Ld^{(j)}_n\bs\Ld^{(j)}_m,b)} & \leq \prod\limits_{k=1}^d \left(\max\left\{ 1+e^{-2|\log\ld_{a,k}|}, 1+e^{-2|\log\ld_{b,k}|}\right\}\right)^{-1}
\end{align*}
which is strictly less than 1.  
It follows:
\begin{align*}
	C(\Ld,ab)&= C(\Ld,a)C(\Ld,b) - D(\Ld)\\
	&= C(\Ld,a)C(\Ld,b)\left(1 - \frac{D(\Ld)}{C(\Ld,a)C(\Ld,b)}\right)\\
	&\geq C(\Ld,a)C(\Ld_b)\left(1 -  \prod\limits_{k=1}^d \left(\max\left\{ 1+e^{-2|\log\tilde{\ld}_{a,k}|}, 1+e^{-2|\log\tilde{\ld}_{b,k}|}\right\}\right)^{-1}\right)
\end{align*}
Substitution of $\tilde{c}$ completes the proof. 
\end{proof}

\begin{lemma}[Diagonal Bound] \label{diagonalbound}
Suppose $\log\tilde{\ld}_{a,j}$ and $\log\tilde{\ld}_{b,k}$ have different signs, then for $n>m$, 
\begin{align*}
	\frac{D(\Ld^{(j)}_n\bs\Ld^{(j)}_m)}{C(\Ld^{(j)}_n\bs\Ld^{(j)}_m,a)C(\Ld^{(j)}_n\bs\Ld^{(j)}_m,b)} \leq (n-m)e^{-2(n-m-1)\max \left\{|\log\tilde{\ld}_{a,j}|,|\log\tilde{\ld}_{b,j}| \right\} } \tag{L3} 
\end{align*}
\end{lemma}
\begin{proof}
Without loss of generality, suppose $\log\tilde{\ld}_{a,j} >0$ and $\log\tilde{\ld}_{b,j} < 0$.  We bound the ratio above by the ratio of sums in the $x_j$ directionby noting that the product of sums in other directions are all bounded above by 1. 
\begin{align*}
	\frac{D(\Ld^{(j)}_n\bs\Ld^{(j)}_m)}{C(\Ld^{(j)}_n\bs\Ld^{(j)}_m,a)C(\Ld^{(j)}_n\bs\Ld^{(j)}_m,b)} &\leq \frac{\sum\limits_{x=m}^{n-1} \tilde{\ld}_{a,j}^{2x}\tilde{\ld}_{b,j}^{2x}}{\sum\limits_{x=m}^{n-1} \tilde{\ld}_{a,j}^{2x}\sum\limits_{x=m}^{n-1} \tilde{\ld}_{b,j}^{2x}}\\
	& = \frac{\ld^{2m}_{a,j}\ld^{2m}_{b,j}\sum\limits_{x=0}^{n-m-1} \tilde{\ld}_{a,j}^{2x}\tilde{\ld}_{b,j}^{2x}}{\ld^{2m}_{a,j}\ld^{2m}_{b,j}\sum\limits_{x=0}^{n-m-1} \tilde{\ld}_{a,j}^{2x}\sum\limits_{x=0}^{n-m-1} \tilde{\ld}_{b,j}^{2x}}\\
	& = \frac{\sum\limits_{x=0}^{n-m-1} \tilde{\ld}_{a,j}^{2x}\tilde{\ld}_{b,j}^{2x}}{\sum\limits_{x=0}^{n-m-1} \tilde{\ld}_{a,j}^{2x}\sum\limits_{x=0}^{n-m-1} \tilde{\ld}_{b,j}^{2x}}\\
	& =\frac{D(\Ld^{(j)}_{n-m-1})}{C(\Ld^{(j)}_{n-m-1},a)C(\Ld^{(j)}_{n-m-1},b)} 
\end{align*}
We substitute $n'={n-m}$ and note that the sum $\sum\limits_{x=0}^{n'-1} \tilde{\ld}_{b,j}^{2x} \geq 1$, so we have
\begin{align*}
	\frac{D(\Ld^{(j)}_{n'})}{C(\Ld^{(j)}_{n'},a)C(\Ld^{(j)}_{n'},b)} \leq \frac{\sum\limits_{x=0}^{n'-1} \tilde{\ld}_{a,j}^{2x}\tilde{\ld}_{b,j}^{2x}}{\sum\limits_{x=0}^{n'-1} \tilde{\ld}_{a,j}^{2x}}\\
\end{align*}
\textit{Case 1: $\tilde{\ld}_{a,j}\tilde{\ld}_{b,j} \geq 1$.}
\begin{align*}	
	\frac{D(\Ld^{(j)}_{n'})}{C(\Ld^{(j)}_{n'},a)C(\Ld^{(j)}_{n'},b)} &\leq \frac{\sum\limits_{x=0}^{n'-1} \tilde{\ld}_{a,j}^{2x}\tilde{\ld}_{b,j}^{2x}}{\sum\limits_{x=0}^{n'-1} \tilde{\ld}_{a,j}^{2x}}\\
	& \leq \frac{n'(\tilde{\ld}_{a,j}\tilde{\ld}_{b,j})^{2(n'-1)}}{\tilde{\ld}_{a,j}^{2n'}}\\
	& \leq n'e^{-2(n'-1)|\log\tilde{\ld}_{b,j}|}
\end{align*}
\textit{Case 2: $\tilde{\ld}_{a,j}\tilde{\ld}_{b,j} \leq 1$.}
\begin{align*}	
	\frac{D(\Ld^{(j)}_{n'})}{C(\Ld^{(j)}_{n'},a)C(\Ld^{(j)}_{n'},b)} &\leq \frac{\sum\limits_{x=0}^{n'-1} \tilde{\ld}_{a,j}^{2x}\tilde{\ld}_{b,j}^{2x}}{\sum\limits_{x=0}^{n'-1} \tilde{\ld}_{a,j}^{2x}}\\
	& \leq \frac{n'}{\tilde{\ld}_{a,j}^{2(n'-1)}}\\
	& \leq n'e^{-2(n'-1)|\log\tilde{\ld}_{a,j}|}
\end{align*}
\end{proof}

\textit{Remark:} The bounds for cases where $\tilde{\ld}_a\tilde{\ld}_b \neq 1$ can be improved to purely exponential bounds.  
The uniformity of bounds in all cases is to decrease the number of cases in later proofs.  
If $\log\tilde{\ld}_a$ and $\log\tilde{\ld}_b$ point in opposite directions, then $\log\tilde{\ld}_{a,j}\log\tilde{\ld}_{b,j} = 1$ and the upper bound in \ref{diagonalbound} will appear no matter the choice of volumes.

The following lemmas show that various ratios of normalization constants are bounded or exponentially small in $\ell$ or $n$ when $\log\tilde{\ld}_{s,j} \neq 0$.   \\
\begin{lemma}[Ratio Lemma] \label{ratiobounds}
For either $s=a$ or $s=b$,
$$\Ld'\subseteq \Ld \Rightarrow C(\Ld',s) \leq C(\Ld,s)$$
If $\tilde{\ld}_{s,j} >1$, 
\begin{align}
	\frac{C(\Ld^{(j)}_{n+1-\ell},s)}{C(\Ld^{(j)}_n,s)} \leq e^{-2(\ell-1)|\log\tilde{\ld}_{s,j}|}\tag{4R1} \label{4R1}\\
	\frac{C(\Ld^{(j)}_{n+1}\bs\Ld^{(j)}_n,s)}{C(\Ld^{(j)}_{n+1}\bs\Ld^{(j)}_{n+1-\ell},s)} \leq 1 \tag{4R2}\label{4R2}\\
	\frac{C(\Ld^{(j)}_{n+1}\bs\Ld^{(j)}_n,s)}{C(\Ld^{(j)}_n,s)} \leq \tilde{\ld}_{s,j}^2 = e^{-2|\log\tilde{\ld}_{s,j}|}\tag{4R3}\label{4R3}\\
 	\frac{C(\Ld^{(j)}_n\bs\Ld^{(j)}_{n+1-\ell},s)}{C(\Ld^{(j)}_n,s)} \leq 1 \tag{4R4} \label{4R4}
\end{align}
If $\tilde{\ld}_{s,j} <1$, 
\begin{align}
	\frac{C(\Ld^{(j)}_{n+1-\ell},s)}{C(\Ld^{(j)}_n,s)} \leq 1 \tag{4L1}\label{4L1}\\
	\frac{C(\Ld^{(j)}_{n+1}\bs\Ld^{(j)}_n,s)}{C(\Ld^{(j)}_{n+1}\bs\Ld^{(j)}_{n+1-\ell},s)} \leq e^{-2(\ell-1)|\log\tilde{\ld}_{s,j}|} \tag{4L2}\label{4L2}\\
	\frac{C(\Ld^{(j)}_{n+1}\bs\Ld^{(j)}_n,s)}{C(\Ld^{(j)}_n,s)} \leq e^{-2(n+1)|\log\tilde{\ld}_{s,j}|} \tag{4L3}\label{4L3}\\
 	\frac{C(\Ld^{(j)}_n\bs\Ld^{(j)}_{n+1-\ell},s)}{C(\Ld^{(j)}_n,s)} \leq 1 \tag{4L4}\label{4L4}
\end{align}

\end{lemma}

\begin{proof}
First, note that if $\Ld' \subseteq \Ld$, then $C(\Ld',s) \leq C(\Ld,s)$ since the latter is the sum of positive terms including those in the former.  
Note that the $\kappa$ cancel in all the ratios as well.  Thus, inequalities \ref{4R2}, \ref{4R4}, \ref{4L1}, and \ref{4L4} hold.
	
The other bounds follow from the ratio of the geometric sums in the $x_j$ direction because the other geometric sums cancel out in the ratios.
For $\tilde{\ld}_{s,j}>1$, inequality \ref{4R1} follows from the following bound:
\begin{align*}
	\frac{C(\Ld^{(j)}_{n+1-\ell},s)}{C(\Ld^{(j)}_n,s)} = \frac{\sum\limits_{x=0}^{n+1-\ell}\tilde{\ld}_{s,j}^{2x}}{\sum\limits_{x=0}^{n}\tilde{\ld}_{s,j}^{2x}}\\
	= \frac{\tilde{\ld}_{s,j}^{2(\ell-1)}\sum\limits_{x=0}^{n+1-\ell}\tilde{\ld}_{s,j}^{2x}}{\tilde{\ld}_{s,j}^{2(\ell-1)}\sum\limits_{x=0}^{n}\tilde{\ld}_{s,j}^{2x}}\\
	= \tilde{\ld}_{s,j}^{-2(\ell-1)}\frac{\sum\limits_{x=\ell-1}^{n}\tilde{\ld}_{s,j}^{2x}}{\sum\limits_{x=0}^{n}\tilde{\ld}_{s,j}^{2x}}\\
	\leq \tilde{\ld}_{s,j}^{-2(\ell-1)}\\
	= e^{-2(\ell-1)|\log\tilde{\ld}_{s,j}|}\\
\end{align*}

Inequality \ref{4R3} follows from the bound $C(\Ld^{(j)}_n,s) \geq \tilde{\ld}^{2n}_{s,j}$:
\begin{align*}
	\frac{C(\Ld^{(j)}_{n+1}\bs\Ld^{(j)}_n,s)}{C(\Ld^{(j)}_n,s)} &\leq \frac{\tilde{\ld}_{s,j}^{2(n+1)}}{\tilde{\ld}_{s,j}^{2n} }\\
	&= \tilde{\ld}_{s,j}^2\\
\end{align*}

For $\tilde{\ld}_{s,j}<1$, inequality \ref{4L2} follows from $C(\Ld^{(j)}_{n+1}\bs\Ld^{(j)}_{n+1-\ell},s) \geq \tilde{\ld}_{s,j}^{2(n+1-\ell+1)}$:
\begin{align*}	
		\frac{C(\Ld^{(j)}_{n+1}\bs\Ld^{(j)}_n,s)}{C(\Ld^{(j)}_{n+1}\bs\Ld^{(j)}_{n+1-\ell},s)} &\leq \frac{\tilde{\ld}_{s,j}^{2(n+1)}}{\tilde{\ld}_{s,j}^{2(n+1-\ell+1)}}\\
		&= \tilde{\ld}_{s,j}^{2(\ell-1)}\\		
		&\leq e^{-2(\ell-1)|\log\tilde{\ld}_{s,j}|}
\end{align*}
Inequality \ref{4L3} follows from $C(\Ld^{(j)}_n,s) \geq 1$:
\begin{align*}
	\frac{C(\Ld^{(j)}_{n+1}\bs\Ld^{(j)}_n,s)}{C(\Ld^{(j)}_n,s)} &\leq \frac{\tilde{\ld}_{s,j}^{2(n+1)}}{1}\\
	&= e^{-2(n+1)|\log\tilde{\ld}_{s,j}|}\\
\end{align*}
\end{proof}

\subsection{Operator Norm Bound}\label{projonb}
The third condition of the martingale method requires a specific bound on the operator norm of the product of the projections $G^{\Ld^{(j)}_{n+1}\bs\Ld^{(j)}_{n+1-\ell}}$ and $E_n$.  This can be thought of as the norm of $G^{\Ld^{(j)}_{n+1}\bs\Ld^{(j)}_{n+1-\ell}}$ acting on any vector in $E_n\mathcal{H}^{\Lambda^{(j)}_{n+1}}$. 
We will drop the $(j)$ from the $\Ld$ notation.

Suppose there is $\ell \geq 2$ such that $\Ld_{n+1}\bs\Ld_{n+1-\ell}$ are connected for $n\geq \ell$.  Vectors $\Psi$ in $E_n\mathcal{H}^{n+1}$ project into the ground state space of $\mathcal{H}^{\Ld_n}$ and are orthogonal to the ground state space of $\mathcal{H}^{\Ld_{n+1}}$. Define $\bd\Lambda_n:=\Ld_{n+1}\bs\Ld_n$. The $\Psi$ in the ground state space of ${H}^{\Ld_n}$ have the general form:
\begin{align*}
\Psi & = \Psi_0^{\Ld_n}\otimes\left(c_{0,0}|0\rangle^{\bd\Lambda_n} + \sum\limits_{x \in \bd\Lambda_n} c_{0,a}(x)|a\rangle_x + \sum\limits_{y \in \bd\Lambda_n} c_{0,b}(y)|b\rangle_y \right. \\
	& \left. + \sum\limits_{x,y \in \bd\Lambda_n,\ x\neq y} c_{0,ab}(x,y)|a\rangle_x|b\rangle_y \right)\\
	& +\Psi_a^{\Ld_n}\otimes\left(c_{a,0}|0\rangle^{\bd\Lambda_n} + \sum\limits_{x \in \bd\Lambda_n} c_{a,a}(x)|a\rangle_x + \sum\limits_{y \in \bd\Lambda_n} c_{a,b}(y)|b\rangle_y \right. \\
	& + \left.	\sum\limits_{x,y \in \bd\Lambda_n,\ x\neq y} c_{a,ab}(x,y)|a\rangle_x|b\rangle_y \right)\\
	& +\Psi_b^{\Ld_n}\otimes\left(c_{b,0}|0\rangle^{\bd\Lambda_n} +
		\sum\limits_{x \in \bd\Lambda_n} c_{b,a}(x)|a\rangle_x \right. \\
		& \left . + \sum\limits_{y \in \bd\Lambda_n} c_{b,b}(y)|b\rangle_y +
		\sum\limits_{x,y \in \bd\Lambda_n,\ x\neq y} c_{b,ab}(x,y)|a\rangle_x|b\rangle_y \right)\\
	& +\Psi_{ab}^{\Ld_n}\otimes\left(c_{ab,0}|0\rangle^{\bd\Lambda_n} +
		\sum\limits_{x \in \bd\Lambda_n} c_{ab,a}(x)|a\rangle_x \right. \\
		&  \left . + \sum\limits_{y \in \bd\Lambda_n} c_{ab,b}(y)|b\rangle_y +
		\sum\limits_{x,y \in \bd\Lambda_n,\ x\neq y} c_{ab,ab}(x,y)|a\rangle_x|b\rangle_y \right)
\end{align*}
which is a linear combination of the tensor product of each of the four ground states of $\mathcal{H}^{\Ld_n}$ with general vectors in $\mathcal{H}^{\bd\Ld_n}$.  
The coefficients of the form $c_{\alpha, \beta}(z)$ correspond to the basis vectors that have the ground state associated with particle species $\alpha$ in $\Lambda_n$ and particles $\beta$ at corresponding positions $z$ in $\bd \Lambda_n$.  
Whenever $\beta \neq 0$, there are particles in $\bd\Lambda_n$ and the corresponding terms are summed over $\bd \Lambda_n$ either single or double sum with exclusion.  
As a convention, $x$ will refer to the position of particle $a$ and $y$ to the position of particle $b$.\\  

Additionally, $\Psi$ must be orthogonal to the ground state space of $\mathcal{H}^{\Lambda_{n+1}}$.  
The vector $\Psi$ is orthogonal to the vacuum ground state when $c_{0,0}=0$.  Orthogonality to $\Psi_a^{n+1}$ requires
\begin{align*}
	\langle \Psi_a^{\Ld_{n+1}}, \Psi\rangle = \frac{1}{\sqrt{C(\Ld_{n+1}, a)}}\sum\limits_{x \in \bd\Lambda_n} c_{0,a}(x)\tilde{\ld}_a^x + \frac{\sqrt{C(\Ld_n,a)}}{\sqrt{C(\Ld_{n+1},a)}}c_{a,0} =0\\
	\Rightarrow c_{a,0}=\frac{-1}{\sqrt{C(\Ld_n,a)}}\sum_{x \in \bd\Lambda_n} c_{0,a}(x)\tilde{\ld}_a^x
\end{align*}
Similarly,
\begin{align*}
	c_{b,0}=\frac{-1}{\sqrt{C(\Ld_n,b)}}\sum_{y \in \bd\Lambda_n}c_{0,b}(y)\tilde{\ld}_b^y
\end{align*}
Orthogonality to $\Psi_{ab}^{n+1}$ requires
\begin{align}
	\langle \Psi_{ab}^{\Ld_{n+1}}, \Psi\rangle = \frac{1}{\sqrt{C(\Ld_{n+1},ab)}} \left[ \sum\limits_{x,y \in \bd\Lambda_n,\ x\neq y} c_{0,ab}(x,y)\tilde{\ld}_a^x\tilde{\ld}_b^y + \sqrt{C(\Ld_n,a)}\sum\limits_{y \in \bd\Lambda_n} c_{a,b}(y)\tilde{\ld}_b^y \right. \nonumber \\
	\left. + \sqrt{C(\Ld_n,b)}\sum\limits_{x \in \bd\Lambda_n} c_{b,a}(x)\tilde{\ld}_a^x + \sqrt{C(\Ld_n,ab)}c_{ab,0} \right] = 0 \nonumber \\
	\Rightarrow c_{ab,0} = \frac{-1}{\sqrt{C(\Ld_n,ab)}}\left[\sum\limits_{x,y \in \bd\Lambda_n,\ x\neq y} c_{0,ab}(x,y) \tilde{\ld}_a^x\tilde{\ld}_b^y \right. \nonumber \\
	\left. + \sqrt{C(\Ld_n,a)}\sum\limits_{y \in \bd\Lambda_n} c_{a,b}(y)\tilde{\ld}_b^y + \sqrt{C(\Ld_n,b)}\sum\limits_{x \in \bd\Lambda_n} c_{b,a}(x)\tilde{\ld}_a^x\right] \label{abortho}
\end{align}
Notation note: we will drop the inputs $(x),\ (y),$ and $(x,y)$ as well as the subscripts from the sums: sums of $c_{*,a}$ are over $x\in\bd\Ld_n$; sums of $c_{*,b}$ are over $y\in\bd\Ld_n$; and sums of $c_{*,ab}$ are over $x,y \in \bd\Lambda_n,\ x\neq y$.  

The projection $G^{\Ld_{n+1}\bs\Ld_{n+1-\ell}}$ is a sum of projections onto the ground states of $\Lambda_{n+1}\backslash \Lambda_{n+1-\ell}$.  The projection of $\Psi$ onto the vacuum ground state of $\Lambda_{n+1}\backslash \Lambda_{n+1-\ell}$ is:
\begin{align*}
	\dyad{\Psi_0} \Psi &= c_{a,0}\frac{\sqrt{C(\Ld_{n+1-\ell},a)}}{\sqrt{C(\Ld_n,a)}}\Psi_a^{\Ld_{n+1-\ell}}\otimes\Psi_0^{\Ld_{n+1}\bs\Ld_{n+1-\ell}}\\
	&+ c_{b,0}\frac{\sqrt{C(\Ld_{n+1-\ell},b)}}{\sqrt{C(\Ld_n,b)}}\Psi_b^{\Ld_{n+1-\ell}}\otimes \Psi_0^{\Ld_{n+1}\bs\Ld_{n+1-\ell}}\\
	&+ c_{ab,0}\frac{\sqrt{C(\Ld_{n+1-\ell}, ab)}}{\sqrt{C(\Ld_n,ab)}}\Psi_{ab}^{\Ld_{n+1-\ell}}\otimes\Psi_0^{\Ld_{n+1}\bs\Ld_{n+1-\ell}}
\end{align*}
The projection of $\Psi$ onto the $a$ ground state of $\Ld_{n+1}\bs\Ld_{n+1-\ell} $ is:
\begin{align*}
	\dyad{\Psi_a} \Psi = \frac{\sum c_{0,a}\tilde{\ld}_a^x}{\sqrt{C(\Ld_{n+1}\bs\Ld_{n+1-\ell},a)}}\Psi_0^{\Ld_{n+1-\ell}}\otimes\Psi_a^{\Ld_{n+1}\bs\Ld_{n+1-\ell}}\\
	+ \frac{C(\Ld_n\bs\Ld_{n+1-\ell},a)c_{a,0}}{\sqrt{C(\Ld_{n+1}\bs\Ld_{n+1-\ell},a)C(\Ld_n,a)}}\Psi_0^{\Ld_{n+1-\ell}}\otimes\Psi_a^{\Ld_{n+1}\bs\Ld_{n+1-\ell}}\\
	+ \frac{\sqrt{C(\Ld_{n+1-\ell},a)}\sum c_{a,a}\tilde{\ld}_a^x}{\sqrt{C(\Ld_{n+1}\bs\Ld_{n+1-\ell},a)C(\Ld_n,a)}}\Psi_a^{\Ld_{n+1-\ell}}\otimes\Psi_a^{\Ld_{n+1}\bs\Ld_{n+1-\ell}}\\
	+ \frac{\sqrt{C(\Ld_{n+1-\ell},b)}\sum c_{b,a}\tilde{\ld}_a^x}{\sqrt{C(\Ld_{n+1}\bs\Ld_{n+1-\ell},a)C(\Ld_n,b)}}\Psi_b^{\Ld_{n+1-\ell}}\otimes\Psi_a^{\Ld_{n+1}\bs\Ld_{n+1-\ell}}\\
	+ \frac{\sqrt{C(\Ld_{n+1-\ell},b)}C(\Ld_n\bs\Ld_{n+1-\ell},a)c_{ab,0}}{\sqrt{C(\Ld_{n+1}\bs\Ld_{n+1-\ell},a)C(\Ld_n,ab)}}\Psi_b^{\Ld_{n+1-\ell}}\otimes\Psi_a^{\Ld_{n+1}\bs\Ld_{n+1-\ell}}\\
	+\frac{\sqrt{C(\Ld_{n+1-\ell},ab)}\sum c_{ab,a}\tilde{\ld}_a^x}{\sqrt{C(\Ld_{n+1}\bs\Ld_{n+1-\ell},a)C(\Ld_n,ab)}}\Psi_{ab}^{\Ld_{n+1-\ell}}\otimes\Psi_a^{\Ld_{n+1}\bs\Ld_{n+1-\ell}}
\end{align*}
The projection of $\Psi$ onto the $b$ ground state of $\Ld_{n+1}\bs\Ld_{n+1-\ell} $ is:
\begin{align*}
	\dyad{\Psi_{b}} \Psi = \frac{\sum c_{0,b}\tilde{\ld}_b^y}{\sqrt{C(\Ld_{n+1}\bs\Ld_{n+1-\ell},b)}}\Psi_0^{\Ld_{n+1-\ell}}\otimes\Psi_b^{\Ld_{n+1}\bs\Ld_{n+1-\ell}}\\
	+ \frac{\sqrt{C(\Ld_{n+1-\ell},a)}\sum c_{a,b}\tilde{\ld}_b^y}{\sqrt{C(\Ld_{n+1}\bs\Ld_{n+1-\ell},b)C(\Ld_n,a)}}\Psi_a^{\Ld_{n+1-\ell}}\otimes\Psi_b^{\Ld_{n+1}\bs\Ld_{n+1-\ell}}\\
	+ \frac{C(\Ld_n\bs\Ld_{n+1-\ell},b)c_{b,0}}{\sqrt{C(\Ld_{n+1}\bs\Ld_{n+1-\ell},b)C(\Ld_n,b)}}\Psi_0^{\Ld_{n+1-\ell}}\otimes\Psi_b^{\Ld_{n+1}\bs\Ld_{n+1-\ell}}\\
	+ \frac{\sqrt{C(\Ld_{n+1-\ell},b)}\sum c_{b,b}\tilde{\ld}_b^y}{\sqrt{C(\Ld_{n+1}\bs\Ld_{n+1-\ell},b)C(\Ld_n,b)}}\Psi_b^{\Ld_{n+1-\ell}}\otimes\Psi_b^{\Ld_{n+1}\bs\Ld_{n+1-\ell}}\\
	+ \frac{\sqrt{C(\Ld_{n+1-\ell},a)}C(\Ld_n\bs\Ld_{n+1-\ell},b)c_{ab,0}}{\sqrt{C(\Ld_{n+1}\bs\Ld_{n+1-\ell},b)C(\Ld_n,ab)}}\Psi_a^{\Ld_{n+1-\ell}}\otimes\Psi_b^{\Ld_{n+1}\bs\Ld_{n+1-\ell}}\\
	+\frac{\sqrt{C(\Ld_{n+1-\ell},ab)}\sum c_{ab,b}\tilde{\ld}_b^y}{\sqrt{C(\Ld_{n+1}\bs\Ld_{n+1-\ell},b)C(\Ld_n,ab)}}\Psi_{ab}^{\Ld_{n+1-\ell}}\otimes\Psi_b^{\Ld_{n+1}\bs\Ld_{n+1-\ell}}
\end{align*}
The projection of $\Psi$ onto the $ab$ ground state of $\Ld_{n+1}\bs\Ld_{n+1-\ell} $ is:
\begin{align*}
	\dyad{\Psi_{ab}} \Psi = \frac{\sum c_{0,ab}\tilde{\ld}_a^x\tilde{\ld}_b^y}{\sqrt{C(\Ld_{n+1}\bs\Ld_{n+1-\ell},ab)}}\Psi_{0}^{\Ld_{n+1-\ell}}\otimes\Psi_{ab}^{\Ld_{n+1}\bs\Ld_{n+1-\ell}}\\
	+\frac{C(\Ld_n\bs\Ld_{n+1-\ell},a)\sum c_{a,b}\tilde{\ld}_b^y}{\sqrt{C(\Ld_{n+1}\bs\Ld_{n+1-\ell},ab)C(\Ld_n,a)}}\Psi_{0}^{\Ld_{n+1-\ell}}\otimes\Psi_{ab}^{\Ld_{n+1}\bs\Ld_{n+1-\ell}}\\
	+\frac{\sqrt{C(\Ld_{n+1-\ell},a)}\sum c_{a,ab}\tilde{\ld}_a^x\tilde{\ld}_b^y}{\sqrt{C(\Ld_{n+1}\bs\Ld_{n+1-\ell},ab)C(\Ld_n,a)}}\Psi_{a}^{\Ld_{n+1-\ell}}\otimes\Psi_{ab}^{\Ld_{n+1}\bs\Ld_{n+1-\ell}}\\
	+\frac{C(\Ld_n\bs\Ld_{n+1-\ell},b) \sum c_{b,a}\tilde{\ld}_a^x}{\sqrt{C(\Ld_{n+1}\bs\Ld_{n+1-\ell},ab)C(\Ld_n,b)}}\Psi_{0}^{\Ld_{n+1-\ell}}\otimes\Psi_{ab}^{\Ld_{n+1}\bs\Ld_{n+1-\ell}}\\
	+\frac{\sqrt{C(\Ld_{n+1-\ell},b)}\sum c_{b,ab}\tilde{\ld}_a^x\tilde{\ld}_b^y}{\sqrt{C(\Ld_{n+1}\bs\Ld_{n+1-\ell},ab)C(\Ld_n,b)}}\Psi_{b}^{\Ld_{n+1-\ell}}\otimes\Psi_{ab}^{\Ld_{n+1}\bs\Ld_{n+1-\ell}}\\
	+\frac{C(\Ld_n\bs\Ld_{n+1-\ell},ab)c_{ab,0}}{\sqrt{C(\Ld_{n+1}\bs\Ld_{n+1-\ell},ab)C(\Ld_n,ab)}}\Psi_{0}^{\Ld_{n+1-\ell}}\otimes\Psi_{ab}^{\Ld_{n+1}\bs\Ld_{n+1-\ell}}\\
	+\frac{C(\Ld_n\bs\Ld_{n+1-\ell},b)\sqrt{C(\Ld_{n+1-\ell},a)}\sum c_{ab,a}\tilde{\ld}_a^x}{\sqrt{C(\Ld_{n+1}\bs\Ld_{n+1-\ell},ab)C(\Ld_n,ab)}}\Psi_{a}^{\Ld_{n+1-\ell}}\otimes\Psi_{ab}^{\Ld_{n+1}\bs\Ld_{n+1-\ell}}\\
	+\frac{C(\Ld_n\bs\Ld_{n+1-\ell},a)\sqrt{C(\Ld_{n+1-\ell},b)}\sum c_{ab,b}\tilde{\ld}_b^y}{\sqrt{C(\Ld_{n+1}\bs\Ld_{n+1-\ell},ab)C(\Ld_n,ab)}}\Psi_{b}^{\Ld_{n+1-\ell}}\otimes\Psi_{ab}^{\Ld_{n+1}\bs\Ld_{n+1-\ell}}\\
	+\frac{\sqrt{C(\Ld_{n+1-\ell},ab)}\sum c_{ab,ab}\tilde{\ld}_a^x\tilde{\ld}_b^y}{\sqrt{C(\Ld_{n+1}\bs\Ld_{n+1-\ell},ab)C(\Ld_n,ab)}}\Psi_{ab}^{\Ld_{n+1-\ell}}\otimes\Psi_{ab}^{\Ld_{n+1}\bs\Ld_{n+1-\ell}}\\
\end{align*}

We compute the value of $\|G^{\Ld_{n+1}\bs\Ld_{n+1-\ell}}\Psi\|^2=\bra{G^{\Ld_{n+1}\bs\Ld_{n+1-\ell}}\Psi}\ket{G^{\Ld_{n+1}\bs\Ld_{n+1-\ell}}\Psi}$ using the orthogonality of the ground states:
\begin{align*}
	\|G^{\Ld^{(j)}_{n+1}\bs\Ld^{(j)}_{n+1-\ell}}\Psi\|^2 = |c_{a,0}|^2\frac{C(\Ld_{n+1-\ell},a)}{C(\Ld_n,a)} + |c_{b,0}|^2\frac{C(\Ld_{n+1-\ell},b)}{C(\Ld_n,b)} + |c_{ab,0}|^2\frac{C(\Ld_{n+1-\ell},ab)}{C(\Ld_n,ab)}\\
	+ \left|\frac{\sum c_{0,a}\tilde{\ld}_a^x}{\sqrt{C(\Ld_{n+1}\bs\Ld_{n+1-\ell},a)}} + \frac{C(\Ld_n\bs\Ld_{n+1-\ell},a)c_{a,0}}{\sqrt{C(\Ld_{n+1}\bs\Ld_{n+1-\ell},a)C(\Ld_n,a)}}\right|^2\\
	+ \frac{C(\Ld_{n+1-\ell},a)}{C(\Ld_{n+1}\bs\Ld_{n+1-\ell},a)C(\Ld_n,a)}\left|\sum c_{a,a} \tilde{\ld}_a^x\right|^2\\
	+ \left|\frac{\sqrt{C(\Ld_{n+1-\ell},b)}\sum c_ {b,a}\tilde{\ld}_a^x}{\sqrt{C(\Ld_{n+1}\bs\Ld_{n+1-\ell},a)C(\Ld_n,b)}} + \frac{C(\Ld_n\bs\Ld_{n+1-\ell},a)\sqrt{C(\Ld_{n+1-\ell},b)}c_{ab,0}}{\sqrt{C(\Ld_{n+1}\bs\Ld_{n+1-\ell},a)C(\Ld_n,ab)}} \right|^2\\
	+ \frac{C(\Ld_{n+1-\ell},ab)}{C(\Ld_{n+1}\bs\Ld_{n+1-\ell},a)C(\Ld_n, ab)}\left|\sum c_{ab,a}\tilde{\ld}_a^x\right|^2\\
	+\left|\frac{\sum c_{0,b}\tilde{\ld}_b^y}{\sqrt{C(\Ld_{n+1}\bs\Ld_{n+1-\ell},b)}} + \frac{C(\Ld_n\bs\Ld_{n+1-\ell},b)c_{b,0}}{\sqrt{C(\Ld_{n+1}\bs\Ld_{n+1-\ell},b)C(\Ld_n,b)}} \right|^2\\
	+\left|\frac{\sqrt{C(\Ld_{n+1-\ell},a)}\sum c_{a,b}\tilde{\ld}_b^y}{\sqrt{C(\Ld_{n+1}\bs\Ld_{n+1-\ell},b)C(\Ld_n,a)}}+\frac{C(\Ld_n\bs\Ld_{n+1-\ell},b)\sqrt{C(\Ld_{n+1-\ell},a)}c_{ab,0}}{\sqrt{C(\Ld_{n+1}\bs\Ld_{n+1-\ell},b)C(\Ld_n,ab)}}   \right|^2 \\
	+\frac{C(\Ld_{n+1-\ell},b)}{C(\Ld_{n+1}\bs\Ld_{n+1-\ell},b)C(\Ld_n,b)}\left|\sum c_{b,b}\tilde{\ld}_b^y\right|^2\\
	+\frac{C(\Ld_{n+1-\ell},ab)}{C(\Ld_{n+1}\bs\Ld_{n+1-\ell},b)C(\Ld_n,ab)}\left|\sum c_{ab,b}\tilde{\ld}_b^y\right|^2\\
	+\left| \frac{\sum c_{0,ab}\tilde{\ld}_a^x\tilde{\ld}_b^y}{\sqrt{C(\Ld_{n+1}\bs\Ld_{n+1-\ell},ab)}} + \frac{C(\Ld_n\bs\Ld_{n+1-\ell},a)\sum c_{a,b}\tilde{\ld}_b^y}{\sqrt{C(\Ld_{n+1}\bs\Ld_{n+1-\ell},ab)C(\Ld_n,a)}}\right. \\
	+\left. \frac{C(\Ld_n\bs\Ld_{n+1-\ell},b)\sum c_{b,a}\tilde{\ld}_a^x}{\sqrt{C(\Ld_{n+1}\bs\Ld_{n+1-\ell},ab)C(\Ld_n,b)}} + \frac{C(\Ld_n\bs\Ld_{n+1-\ell},ab) c_{ab,0}}{\sqrt{C(\Ld_{n+1}\bs\Ld_{n+1-\ell},ab)C(\Ld_n,ab)}}\right|^2\\
	+\left|\frac{\sqrt{C(\Ld_{n+1-\ell},a)}\sum c_{a,ab}\tilde{\ld}_a^x\tilde{\ld}_b^y}{\sqrt{C(\Ld_{n+1}\bs\Ld_{n+1-\ell},ab)C(\Ld_n,a)}}+\frac{C(\Ld_n\bs\Ld_{n+1-\ell},b)\sqrt{C(\Ld_{n+1-\ell},a)}\sum c_{ab,a}\tilde{\ld}_a^x}{\sqrt{C(\Ld_{n+1}\bs\Ld_{n+1-\ell},ab)C(\Ld_n,ab)}}\right|^2\\
	+\left|\frac{\sqrt{C(\Ld_{n+1-\ell},b)}\sum c_{b,ab}\tilde{\ld}_a^x\tilde{\ld}_b^y}{\sqrt{C(\Ld_{n+1}\bs\Ld_{n+1-\ell},ab)C(\Ld_n,b)}}+\frac{C(\Ld_n\bs\Ld_{n+1-\ell},a)\sqrt{C(\Ld_{n+1-\ell},b)}\sum c_{ab,b}\tilde{\ld}_b^y}{\sqrt{C(\Ld_{n+1}\bs\Ld_{n+1-\ell},ab)C(\Ld_n,ab)}}\right|^2\\
	+\frac{C(\Ld_{n+1-\ell},ab)|\sum c_{ab,ab}\tilde{\ld}_a^x\tilde{\ld}_b^y|^2}{C(\Ld_n,ab)C(\Ld_{n+1}\bs\Ld_{n+1-\ell},ab)}
\end{align*}
We will decompose the sum of terms into terms associated with the projection onto subspaces separately.
We apply the following inequalities which bound the $\ell^1$ metric by the $\ell^2$ metric in two and three dimensions, respectively and will appear in the following form:
\begin{align}
	&|x+y|^2 \leq 2(|x|^2+|y|^2) \label{2ball} \\
	&|x+y+z|^2 \leq 3(|x|^2+|y|^2+|z|^2) \label{3ball}
\end{align} 
as well as the Cauchy-Schwarz identity in the following form: 
	\begin{align*}
	&\left|\sum\limits_x c(x)\ld_s^x \right|^2 \leq \sum\limits_x |c(x)|^2 \sum\limits_x \ld_s^{2x} \label{CS} &\text{   (Cauchy-Schwarz inequality)}
\
\end{align*}

The norm-squared of projection onto space with only particles of species $a$ is the same as the single particle PVBS model and was calculated in \cite{bishop2016spectral}.
\begin{align*}
	|c_{a,0}|^2\frac{C(\Ld_{n+1-\ell},a)}{C(\Ld_n,a)} + \frac{C(\Ld_{n+1-\ell},a)}{C(\Ld_{n+1}\bs\Ld_{n+1-\ell},a)C(\Ld_n,a)}\left|\sum c_{a,a} \tilde{\ld}_a^x\right|^2\\
	+  \left|\frac{\sum c_{0,a}\tilde{\ld}_a^x}{\sqrt{C(\Ld_{n+1}\bs\Ld_{n+1-\ell},a)}} + \frac{C(\Ld_n\bs\Ld_{n+1-\ell},a)c_{a,0}}{\sqrt{C(\Ld_{n+1}\bs\Ld_{n+1-\ell},a)C(\Ld_n,a)}}\right|^2\\
\leq \frac{C(\Ld_{n+1-\ell},a)}{C(\Ld_n,a)}\frac{C(\Ld_{n+1}\bs\Ld_n , a)}{C(\Ld_{n+1}\bs\Ld_{n+1-\ell},a)}\left(|c_{a,0}|^2 + \sum |c_{0,a}|^2 + |c_{a,a}|^2\right)
\end{align*}
where the term $|c_{a,0}|^2 + \sum |c_{0,a}|^2 + |c_{a,a}|^2$ is the norm of $\Psi$ projected into the $a$ particle subspace.
If $\tilde{\ld}_{a,j} >1$, the first ratio is bounded by $e^{-2(\ell-1)|\log\tilde{\ld}_{a,j}|}$ by (\ref{4R1}) and the second ratio is bounded by 1 by (\ref{4R2}).
If $\tilde{\ld}_{a,j}<1$, then the first ratio is bounded by $1$ by (\ref{4L1}) and the second ratio is bounded by $e^{-2(\ell-1)|\log\tilde{\ld}_{a,j}|}$ by (\ref{4L2}).  
Therefore,

\begin{align*}
\left(\frac{C(\Ld_{n+1-\ell},a)}{C(\Ld_n,a)}\right)\left(\frac{C(\Ld_{n+1}\bs\Ld_n , a)}{C(\Ld_{n+1}\bs\Ld_{n+1-\ell},a)}\right)\left(|c_{a,0}|^2 + \sum |c_{0,a}|^2 + |c_{a,a}|^2\right)\\
\leq e^{-2(\ell-1)|\log\tilde{\ld}_{a,j}|}\left(|c_{a,0}|^2 + \sum |c_{0,a}|^2 + |c_{a,a}|^2\right)
\end{align*}
The operator norm squared of the projections is bounded above by $e^{-2(\ell-1)|\log\tilde{\ld}_{a,j}|}$ when restricted the a-species only subspace.  
Similarly, the projection onto space with only particles of species $b$:  
\begin{align*}
	|c_{b,0}|^2\frac{C(\Ld_{n+1-\ell},b)}{C(\Ld_n,b)} + \frac{C(\Ld_{n+1-\ell},b)}{C(\Ld_{n+1}\bs\Ld_{n+1-\ell},b)C(\Ld_n,b)}\left|\sum c_{b,b} \tilde{\ld}_b^x\right|^2\\
	+  \left|\frac{\sum c_{0,b}\tilde{\ld}_b^x}{\sqrt{C(\Ld_{n+1}\bs\Ld_{n+1-\ell},b)}} + \frac{C(\Ld_n\bs\Ld_{n+1-\ell},b)c_{b,0}}{\sqrt{C(\Ld_{n+1}\bs\Ld_{n+1-\ell},b)C(\Ld_n,b)}}\right|^2\\
= \frac{C(\Ld_{n+1-\ell},b)}{C(\Ld_n,b)}\frac{C(\Ld_{n+1}\bs\Ld_n , b)}{C(\Ld_{n+1}\bs\Ld_{n+1-\ell},b)}\left(|c_{b,0}|^2 + \sum |c_{0,b}|^2 + |c_{b,b}|^2\right)\\
	\leq e^{-2(\ell-1)|\log\tilde{\ld}_{b,j}|}\left(|c_{b,0}|^2 + \sum |c_{0,b}|^2 + |c_{b,b}|^2\right)
\end{align*}
The operator norm squared of the projections is bounded above by $e^{-2(\ell-1)|\log\tilde{\ld}_{b,j}|}$ when restricted the b-species only subspace.

The norm squared of the projection of $G^{\Ld_{n+1}\bs\Ld_{n+1-\ell}}\Psi$ onto the space with three or more particles is 
\begin{align*}
	 \frac{C(\Ld_{n+1-\ell},ab)}{C(\Ld_{n+1}\bs\Ld_{n+1-\ell},a)C(\Ld_n, ab)}\left|\sum c_{ab,a}\tilde{\ld}_a^x\right|^2\\
	+\left|\frac{\sqrt{C(\Ld_{n+1-\ell},a)}\sum c_{a,ab}\tilde{\ld}_a^x\tilde{\ld}_b^y}{\sqrt{C(\Ld_{n+1}\bs\Ld_{n+1-\ell},ab)C(\Ld_n,a)}}+\frac{C(\Ld_n\bs\Ld_{n+1-\ell},b)\sqrt{C(\Ld_{n+1-\ell},a)}\sum c_{ab,a}\tilde{\ld}_a^x}{\sqrt{C(\Ld_{n+1}\bs\Ld_{n+1-\ell},ab)C(\Ld_n,ab)}}\right|^2\\
	+\frac{C(\Ld_{n+1-\ell},ab)}{C(\Ld_{n+1}\bs\Ld_{n+1-\ell},b)C(\Ld_n,ab)}\left|\sum c_{ab,b}\tilde{\ld}_b^y\right|^2\\
	+\left|\frac{\sqrt{C(\Ld_{n+1-\ell},b)}\sum c_{b,ab}\tilde{\ld}_a^x\tilde{\ld}_b^y}{\sqrt{C(\Ld_{n+1}\bs\Ld_{n+1-\ell},ab)C(\Ld_n,b)}}+\frac{C(\Ld_n\bs\Ld_{n+1-\ell},a)\sqrt{C(\Ld_{n+1-\ell},b)}\sum c_{ab,b}\tilde{\ld}_b^y}{\sqrt{C(\Ld_{n+1}\bs\Ld_{n+1-\ell},ab)C(\Ld_n,ab)}}\right|^2\\
	+\frac{C(\Ld_{n+1-\ell},ab)|\sum c_{ab,ab}\tilde{\ld}_a^x\tilde{\ld}_b^y|^2}{C(\Ld_n,ab)C(\Ld_{n+1}\bs\Ld_{n+1-\ell},ab)}
\end{align*}
Applying the inequality $|x+y|^2 \leq 2|x|^2+2|y|^2$ as well as Cauchy-Schwarz inequality on $|\sum c_* \tilde{\ld}_*|^2$: we calculate an upper bound
\begin{align*}
	 \frac{C(\Ld_{n+1-\ell},ab)C(\Ld_{n+1}\bs\Ld_{n},a)}{C(\Ld_{n+1}\bs\Ld_{n+1-\ell},a)C(\Ld_n, ab)}\sum |c_{ab,a}|^2\\
	+ 2\frac{C(\Ld_{n+1-\ell},a)C(\Ld_{n+1}\bs\Ld_n,ab)}{C(\Ld_{n+1}\bs\Ld_{n+1-\ell},ab)C(\Ld_n,a)}\sum |c_{a,ab}|^2\\
	+ 2\frac{C^2(\Ld_n\bs\Ld_{n+1-\ell},b)C(\Ld_{n+1-\ell},a)C(\Ld_{n+1}\bs\Ld_n,a)}{C(\Ld_{n+1}\bs\Ld_{n+1-\ell},ab)C(\Ld_n,ab)} \sum |c_{ab,a}|^2\\
	+\frac{C(\Ld_{n+1-\ell},ab)C(\Ld_{n+1}\bs\Ld_n,b)}{C(\Ld_{n+1}\bs\Ld_{n+1-\ell},b)C(\Ld_n,ab)}\sum |c_{ab,b}|^2\\
	+2 \frac{C(\Ld_{n+1-\ell},b)C(\Ld_{n+1}\bs\Ld_n,ab)}{C(\Ld_{n+1}\bs\Ld_{n+1-\ell},ab)C(\Ld_n,b)}\sum |c_{b,ab}|^2\\
	+2\frac{C^2(\Ld_n\bs\Ld_{n+1-\ell},a)C(\Ld_{n+1-\ell},b)C(\Ld_{n+1}\bs\Ld_n,b)}{C(\Ld_{n+1}\bs\Ld_{n+1-\ell},ab)C(\Ld_n,ab)}  \sum |c_{ab,b}|^2\\
	+\frac{C(\Ld_{n+1-\ell},ab)C(\Ld_{n+1}\bs\Ld_n,ab)}{C(\Ld_n,ab)C(\Ld_{n+1}\bs\Ld_{n+1-\ell},ab)}   \sum |c_{ab,ab}|^2
\end{align*}
We bound the $C(\Ld,ab)$ by $C(\Ld,a)C(\Ld,b)$ using Lemma (\ref{productbounds}) and group by norm on projected particle subspaces:
\begin{align*}
	 \leq \left[\tilde{c}\frac{C(\Ld_{n+1-\ell},a)C(\Ld_{n+1}\bs\Ld_{n},a)}{C(\Ld_{n+1}\bs\Ld_{n+1-\ell},a)C(\Ld_n, a)} + 2\tilde{c}^2\frac{C(\Ld_{n+1-\ell},a)C(\Ld_{n+1}\bs\Ld_n,a)}{C(\Ld_{n+1}\bs\Ld_{n+1-\ell},a)C(\Ld_n,a)}    \right]\sum |c_{ab,a}|^2\\
	 + \left[\tilde{c}\frac{C(\Ld_{n+1-\ell},b)C(\Ld_{n+1}\bs\Ld_n,b)}{C(\Ld_{n+1}\bs\Ld_{n+1-\ell},b)C(\Ld_n,b)} +2\tilde{c}^2\frac{C(\Ld_{n+1-\ell},b)C(\Ld_{n+1}\bs\Ld_n,b)}{C(\Ld_{n+1}\bs\Ld_{n+1-\ell},b)C(\Ld_n,b)} 	\right]\sum |c_{ab,b}|^2\\
	 + 2\tilde{c}\frac{C(\Ld_{n+1-\ell},a)C(\Ld_{n+1}\bs\Ld_n,a)}{C(\Ld_{n+1}\bs\Ld_{n+1-\ell},a)C(\Ld_n,a)}\sum |c_{a,ab}|^2\\
	 +2\tilde{c}\frac{C(\Ld_{n+1-\ell},b)C(\Ld_{n+1}\bs\Ld_n,b)}{C(\Ld_{n+1}\bs\Ld_{n+1-\ell},b)C(\Ld_n,b)}\sum |c_{b,ab}|^2\\
	+\tilde{c}^2\frac{C(\Ld_{n+1-\ell},a)C(\Ld_{n+1}\bs\Ld_n,a)}{C(\Ld_n,a)C(\Ld_{n+1}\bs\Ld_{n+1-\ell},a)} \frac{C(\Ld_{n+1-\ell},b)C(\Ld_{n+1}\bs\Ld_n,b)}{C(\Ld_n,b)C(\Ld_{n+1}\bs\Ld_{n+1-\ell},b)}   \sum |c_{ab,ab}|^2
\end{align*}
We group the ratios of normalization constants by particle type and so each set appearing in the numerator is a subset of the set in the denominator.
We factor the ratios of normalization constants by particle type (a, then b) then by their bound dependent on $\ld_s$. 
\begin{align*}
	 \leq \left[\tilde{c} \left(\frac{C(\Ld_{n+1-\ell},a)}{C(\Ld_n, a)}\right)\left(\frac{C(\Ld_{n+1}\bs\Ld_{n},a)}{C(\Ld_{n+1}\bs\Ld_{n+1-\ell},a)}\right) \right.\\
	 \left. + 2\tilde{c}^2\left(\frac{C(\Ld_{n+1-\ell},a)}{C(\Ld_n,a)}\right)\left(\frac{C(\Ld_{n+1}\bs\Ld_n,a)}{C(\Ld_{n+1}\bs\Ld_{n+1-\ell},a)}\right)\right]\sum |c_{ab,a}|^2\\
	 	 + \left[\tilde{c} \left(\frac{C(\Ld_{n+1-\ell},b)}{C(\Ld_n, b)}\right)\left(\frac{C(\Ld_{n+1}\bs\Ld_{n},b)}{C(\Ld_{n+1}\bs\Ld_{n+1-\ell},b)}\right) \right.\\
	 \left. + 2\tilde{c}^2\left(\frac{C(\Ld_{n+1-\ell},b)}{C(\Ld_n,b)}\right)\left(\frac{C(\Ld_{n+1}\bs\Ld_n,b)}{C(\Ld_{n+1}\bs\Ld_{n+1-\ell},b)}\right)\right]\sum |c_{ab,b}|^2\\
	 + 2\tilde{c}\left(\frac{C(\Ld_{n+1-\ell},a)}{C(\Ld_n,a)}\right)\left(\frac{C(\Ld_{n+1}\bs\Ld_n,a)}{C(\Ld_{n+1}\bs\Ld_{n+1-\ell},a)}\right)\sum |c_{a,ab}|^2\\
		 + 2\tilde{c}\left(\frac{C(\Ld_{n+1-\ell},b)}{C(\Ld_n,b)}\right)\left(\frac{C(\Ld_{n+1}\bs\Ld_n,b)}{C(\Ld_{n+1}\bs\Ld_{n+1-\ell},b)}\right)\sum |c_{b,ab}|^2\\
	+\tilde{c}^2\left(\frac{C(\Ld_{n+1-\ell},a)}{C(\Ld_n,a)}\right)\left(\frac{C(\Ld_{n+1}\bs\Ld_n,a)}{C(\Ld_{n+1}\bs\Ld_{n+1-\ell},a)}\right)\\
	\cdot\left(\frac{C(\Ld_{n+1-\ell},b)}{C(\Ld_n,b)}\right)\left(\frac{C(\Ld_{n+1}\bs\Ld_n,b)}{C(\Ld_{n+1}\bs\Ld_{n+1-\ell},b)}  \right) \sum |c_{ab,ab}|^2
\end{align*}
In each pair of ratios, the first term is bounded by $e^{-2(\ell-1)|\log\tilde{\ld}_{s,j}|}$ when $\ld_s>1$ (\ref{4R1}) and bounded by 1 when $\ld_s<1$ (\ref{4R2}) and the second term is bounded by 1 when $\ld_s<1$ (\ref{4L1}) and by $e^{-2(\ell-1)|\log\tilde{\ld}_{s,j}|}$ when $\ld_s<1$ (\ref{4L2}).  
Each pair of ratios of normalization constants is bounded above by the maximum of $\exp[-2(\ell-1)|\log\tilde{\ld}_s(j)|]$ for $s =a,\ b$.  
Each of the sums on the right side are projections of $\Psi$ onto subspaces of specific configurations of three or four particles.  
We take the maximum of the possible coefficients to obtain an upper bound on the norm squared on the particle subspace with three or more particles:
\begin{align} 
	3\tilde{c}^2\max\{\exp[-2(\ell-1)|\log\tilde{\ld}_{s,j}|]:s=a,b\}
\end{align}
where the terms $\sum|c_{ab,a}|^2+\sum|c_{ab,b}|^2+\sum|c_{a,ab}|^2+\sum|c_{b,ab}|^2+\sum|c_{ab,ab}|^2$ is the norm-squared of the projection of $\Psi$ onto this subspace.  

For the terms in subspace with exactly one of each particle species, the norm squared is:
\begin{align*}
	|c_{ab,0}|^2\frac{C(\Ld_{n+1-\ell},ab)}{C(\Ld_n,ab)} \\
	+ \left|\frac{\sqrt{C(\Ld_{n+1-\ell},b)}\sum c_ {b,a}\tilde{\ld}_a^x}{\sqrt{C(\Ld_{n+1}\bs\Ld_{n+1-\ell},a)C(\Ld_n,b)}} + \frac{C(\Ld_n\bs\Ld_{n+1-\ell},a)\sqrt{C(\Ld_{n+1-\ell},b)}c_{ab,0}}{\sqrt{C(\Ld_{n+1}\bs\Ld_{n+1-\ell},a)C(\Ld_n,ab)}} \right|^2\\
	+\left|\frac{\sqrt{C(\Ld_{n+1-\ell},a)}\sum c_{a,b}\tilde{\ld}_b^y}{\sqrt{C(\Ld_{n+1}\bs\Ld_{n+1-\ell},b)C(\Ld_n,a)}}+\frac{C(\Ld_n\bs\Ld_{n+1-\ell},b)\sqrt{C(\Ld_{n+1-\ell},a)}c_{ab,0}}{\sqrt{C(\Ld_{n+1}\bs\Ld_{n+1-\ell},b)C(\Ld_n,ab)}}   \right|^2 \\
+\left| \frac{\sum c_{0,ab}\tilde{\ld}_a^x\tilde{\ld}_b^y}{\sqrt{C(\Ld_{n+1}\bs\Ld_{n+1-\ell},ab)}} + \frac{C(\Ld_n\bs\Ld_{n+1-\ell},a)\sum c_{a,b}\tilde{\ld}_b^y}{\sqrt{C(\Ld_{n+1}\bs\Ld_{n+1-\ell},ab)C(\Ld_n,a)}}\right. \\
	+\left. \frac{C(\Ld_n\bs\Ld_{n+1-\ell},b)\sum c_{b,a}\tilde{\ld}_a^x}{\sqrt{C(\Ld_{n+1}\bs\Ld_{n+1-\ell},ab)C(\Ld_n,b)}} + \frac{C(\Ld_n\bs\Ld_{n+1-\ell},ab) c_{ab,0}}{\sqrt{C(\Ld_{n+1}\bs\Ld_{n+1-\ell},ab)C(\Ld_n,ab)}}\right|^2
\end{align*}
We apply the orthogonality condition (\ref{abortho}) on all terms of the form $c_{ab,0}$ to obtain
\begin{align*}
	= \frac{C(\Ld_{n+1-\ell},ab)}{C^2(\Ld_n,ab)}\left| \sum c_{0,ab}\tilde{\ld}_a^x\tilde{\ld}_b^y + \sqrt{C(\Ld_n,a)}\sum c_{a,b}\tilde{\ld}_b^y + \sqrt{C(\Ld_n,b)}\sum c_{b,a} \tilde{\ld}_a^x\right|^2\\
	+ \frac{C(\Ld_{n+1-\ell},b)}{C(\Ld_{n+1}\bs\Ld_{n+1-\ell},a)}\left|\frac{C(\Ld_n,ab) - C(\Ld_n,b)C(\Ld_n\bs\Ld_{n+1-\ell},a)}{C(\Ld_n,ab)\sqrt{C(\Ld_n,b)}}\sum c_{b,a}\tilde{\ld}_a^x \right.\\
	- \left.\frac{C(\Ld_n\bs\Ld_{n+1-\ell},a)}{C(\Ld_n,ab)}\left(\sum c_{0,ab}\tilde{\ld}_a^x\tilde{\ld}_b^y + \sqrt{C(\Ld_n,a)}\sum c_{a,b}\tilde{\ld}_b^y\right) \right|^2\\
	+ \frac{C(\Ld_{n+1-\ell},a)}{C(\Ld_{n+1}\bs\Ld_{n+1-\ell},b)}\left|\frac{C(\Ld_n,ab) - C(\Ld_n,a)C(\Ld_n\bs\Ld_{n+1-\ell},b)}{C(\Ld_n,ab)\sqrt{C(\Ld_n,a)}}\sum c_{a,b}\tilde{\ld}_b^y \right.\\
	- \left.\frac{C(\Ld_n\bs\Ld_{n+1-\ell},b)}{C(\Ld_n,ab)}\left(\sum c_{0,ab}\tilde{\ld}_a^x\tilde{\ld}_b^y + \sqrt{C(\Ld_n,b)}\sum c_{b,a}\tilde{\ld}_a^x\right) \right|^2\\
	+ \frac{1}{C(\Ld_{n+1}\bs\Ld_{n+1-\ell},ab)}\left| \frac{C(\Ld_n,ab) - C(\Ld_n\bs\Ld_{n+1-\ell},ab)}{C(\Ld_n,ab)}\sum c_{0,ab}\tilde{\ld}_a^x\tilde{\ld}_b^y\right.\\
	+ \left. \frac{C(\Ld_n\bs\Ld_{n+1-\ell},a)C(\Ld_n,ab) - C(\Ld_n\bs\Ld_{n+1-\ell},ab)C(\Ld_n,a)}{C(\Ld_n, ab)\sqrt{C(\Ld_n,a)}}\sum c_{a,b}\tilde{\ld}_b^y\right.\\
	+ \left. \frac{C(\Ld_n\bs\Ld_{n+1-\ell},b)C(\Ld_n,ab) - C(\Ld_n\bs\Ld_{n+1-\ell},ab)C(\Ld_n,b)}{C(\Ld_n, ab)\sqrt{C(\Ld_n,b)}}\sum c_{b,a}\tilde{\ld}_a^x\right|^2\\
\end{align*}
and bound using $|x+y+z|^2 \leq 3(|x|^2+|y|^2 + |z|^2)$:
\begin{align}  
	\leq 3\left|\sum c_{b,a}\tilde{\ld}_a^x\right|^2\left[ \frac{C(\Ld_{n+1-\ell},ab)C(\Ld_n,b)}{C^2(\Ld_n,ab)} \right.\nonumber\\
	\left. + \frac{C(\Ld_{n+1-\ell},b)\left|C(\Ld_n,ab)-C(\Ld_n,b)C(\Ld_n\bs\Ld_{n+1-\ell},a)\right|^2}{C(\Ld_{n+1}\bs\Ld_{n+1-\ell},a)C^2(\Ld_n,ab)C(\Ld_n,b)}\right.\nonumber\\
	\left. + \frac{C(\Ld_{n+1-\ell},a)C^2(\Ld_n\bs\Ld_{n+1-\ell},b)C(\Ld_n,b)}{C^2(\Ld_n,ab)C(\Ld_{n+1}\bs\Ld_{n+1-\ell},b)}\right.\nonumber\\
	\left. + \frac{\left|C(\Ld_n\bs\Ld_{n+1-\ell},b)C(\Ld_n,ab)-C(\Ld_n\bs\Ld_{n+1-\ell},ab)C(\Ld_n,b)\right|^2}{C(\Ld_{n+1}\bs\Ld_{n+1-\ell},ab)C^2(\Ld_n,ab)C(\Ld_n,b)}\right]\nonumber\\
	+ 3\left|\sum c_{a,b}\tilde{\ld}_b^y\right|^2\left[\frac{ C(\Ld_{n+1-\ell},ab)C(\Ld_n,a) }{C^2(\Ld_n,ab)}\right.\nonumber\\
	\left. + \frac{C(\Ld_{n+1-\ell},b)C^2(\Ld_n\bs\Ld_{n+1-\ell},a)C(\Ld_n,a)}{C(\Ld_{n+1}\bs\Ld_{n+1-\ell},a)C^2(\Ld_n,ab)}\right.\nonumber\\
	\left. + \frac{C(\Ld_{n+1-\ell},a)\left|C(\Ld_n,ab)-C(\Ld_n,a)C(\Ld_n\bs\Ld_{n+1-\ell},b)\right|^2}{C(\Ld_{n+1}\bs\Ld_{n+1-\ell},b)C^2(\Ld_n,ab)C(\Ld_n,a)}\right.\nonumber\\
	\left. + \frac{\left|C(\Ld_n\bs\Ld_{n+1-\ell},a)C(\Ld_n,ab)-C(\Ld_n\bs\Ld_{n+1-\ell},ab)C(\Ld_n,a)\right|^2}{C(\Ld_{n+1}\bs\Ld_{n+1-\ell},ab)C^2(\Ld_n,ab)C(\Ld_n,a)}\right]\nonumber\\
	+ 3\left|\sum c_{0,ab}\tilde{\ld}_a^x\tilde{\ld}_b^y\right|^2\left[\frac{C(\Ld_{n+1-\ell},ab)}{C^2(\Ld_n,ab)} \right. + \frac{C(\Ld_{n+1-\ell},b)C^2(\Ld_n\bs\Ld_{n+1-\ell},a)}{C(\Ld_{n+1}\bs\Ld_{n+1-\ell},a)C^2(\Ld_n,ab)} \nonumber\\
	+\left.  \frac{C(\Ld_{n+1-\ell},a)C^2(\Ld_n\bs\Ld_{n+1-\ell},b)}{C(\Ld_{n+1}\bs\Ld_{n+1-\ell},b)C^2(\Ld_n,ab)} + \frac{\left|C(\Ld_n,ab)-C(\Ld_n\bs\Ld_{n+1-\ell},ab)\right|^2}{C(\Ld_{n+1}\bs\Ld_{n+1-\ell},ab)C^2(\Ld_n,ab)}\right]	
\end{align}

We will bound the terms associated with $c_{b,a}$ first and obtain a bound on the $c_{a,b}$ by interchanging $a$ and $b$.
We rewrite the terms in the absolute value sign in the second line using
\begin{align*}
	 &\mathbf{C(\Ld_n,ab) - C(\Ld_n,b) C(\Ld_n\bs\Ld_{n+1-\ell},a)}\\
	&\ \ = C(\Ld_n,a)C(\Ld_n,b)-D(\Ld_n) -  C(\Ld_n\bs\Ld_{n+1-\ell},a)C(\Ld_n,b)\\
	&\ \ = \left(C(\Ld_n,a) - C(\Ld_n\bs\Ld_{n+1-\ell},a)\right) C(\Ld_n,b) -D(\Ld_n)\\
	&\ \ \mathbf{= C(\Ld_{n+1-\ell},a)C(\Ld_n,b) -D(\Ld_n)}
\end{align*}
and in fourth line using 
\begin{align*}		 
	&\mathbf{C(\Ld_n\bs\Ld_{n+1-\ell},b)C(\Ld_n, ab) - C(\Ld_n\bs\Ld_{n+1-\ell},ab) C(\Ld_n, b)} \\
	&= C(\Ld_n\bs\Ld_{n+1-\ell},b)[ C(\Ld_n, a)C(\Ld_n, b) - D(\Ld_n)] \\ 
	&\ \ - [ C(\Ld_n\bs\Ld_{n+1-\ell},a)C(\Ld_n\bs\Ld_{n+1-\ell},b) - D(\Ld_n\bs\Ld_{n+1-\ell})] C(\Ld_n, b) \\
	&= C(\Ld_n\bs\Ld_{n+1-\ell},b)C(\Ld_n,b)[ C(\Ld_n,a)-C(\Ld_n\bs\Ld_{n+1-\ell},a)]\\
	&\ \ - D(\Ld_n)C(\Ld_n\bs\Ld_{n+1-\ell},b) + D(\Ld_n\bs\Ld_{n+1-\ell})C(\Ld_n, b) \\
	&= C(\Ld_n\bs\Ld_{n+1-\ell},b)C(\Ld_n,b)C(\Ld_{n+1-\ell},a) - D(\Ld_n)C(\Ld_n\bs\Ld_{n+1-\ell},b)\\
	&\ \  + D(\Ld_n\bs\Ld_{n+1-\ell})[C(\Ld_n\bs\Ld_{n+1-\ell}, b) - C(\Ld_{n+1-\ell},b)] \\
		&= C(\Ld_n\bs\Ld_{n+1-\ell},b)C(\Ld_n,b)C(\Ld_{n+1-\ell},a) + D(\Ld_n\bs\Ld_{n+1-\ell})C(\Ld_{n+1-\ell},b) \\
	&\ \ -[D(\Ld_n)	- D(\Ld_n\bs\Ld_{n+1-\ell})]C(\Ld_n\bs\Ld_{n+1-\ell},b) \\
	& \mathbf{=C(\Ld_{n+1-\ell},a)C(\Ld_n,b)C(\Ld_n\bs\Ld_{n+1-\ell},b)}\\ 
	&\ \ \mathbf{+ D(\Ld_n\bs\Ld_{n+1-\ell})C(\Ld_{n+1-\ell},b) - D(\Ld_{n+1-\ell})C(\Ld\bs\Ld_{n+1-\ell},b)}.
\end{align*} 
We bound above using the above substitutions, expanding the squared terms using the bounds \ref{2ball} and \ref{3ball}, the product bounds on  $C(\Ld,ab)$, and the Cauchy-Schwarz inequality:
$$\left|\sum\limits_{x \in \bd\Ld_n} c_{b,a}\tilde{\ld}^{2x}_a(j)\right|^2 \leq \sum\limits_{x \in \bd\Ld_n}|c_{b,a}|^2 C(\Ld_{n+1}\bs\Ld_n,a)$$  
We have the upper bound
\begin{align*}
	\leq 3\sum |c_{b,a}|^2\left[ \tilde{c}^2\frac{C(\Ld_{n+1-\ell},a)C(\Ld_{n+1}\bs\Ld_n,a)C(\Ld_{n+1-\ell},b)C(\Ld_n,b)}{C^2(\Ld_n,a)C^2(\Ld_n,b)} \right.\\
	\left. + 2\tilde{c}^2\frac{C(\Ld_{n+1}\bs\Ld_n,a)C(\Ld_{n+1-\ell},b)C^2(\Ld_{n+1-\ell},a)C^2(\Ld_n,b)}{C^2(\Ld_n,a)C(\Ld_{n+1}\bs\Ld_{n+1-\ell},a)C^3(\Ld_n,b)}\right.\\
	\left. + 2\tilde{c}^2\frac{C(\Ld_{n+1}\bs\Ld_n,a)C(\Ld_{n+1-\ell},b)D^2(\Ld_n)}{C^2(\Ld_n,a)C(\Ld_{n+1}\bs\Ld_{n+1-\ell},a)C^3(\Ld_n,b)}\right.\\
	\left. + \tilde{c}^2\frac{C(\Ld_{n+1}\bs\Ld_n,a)C(\Ld_{n+1-\ell},a)C^2(\Ld_n\bs\Ld_{n+1-\ell},b)C(\Ld_n,b)}{C^2(\Ld_n,a)C^2(\Ld_n,b)C(\Ld_{n+1}\bs\Ld_{n+1-\ell},b)}\right.\\
	\left. + 3\tilde{c}^3\frac{C(\Ld_{n+1}\bs\Ld_n,a)C^2(\Ld_{n+1-\ell},a)C^2(\Ld_n,b)C^2(\Ld_n\bs\Ld_{n+1-\ell},b)}{C^2(\Ld_n,a)C(\Ld_{n+1}\bs\Ld_{n+1-\ell},a)C^3(\Ld_n,b)C(\Ld_{n+1}\bs\Ld_{n+1-\ell},b)}\right.\\
	\left. + 3\tilde{c}^3\frac{C(\Ld_{n+1}\bs\Ld_n,a)D^2(\Ld_n\bs\Ld_{n+1-\ell})C^2(\Ld_{n+1-\ell},b)}{C^2(\Ld_n,a)C(\Ld_{n+1}\bs\Ld_{n+1-\ell},a)C^3(\Ld_n,b)C(\Ld_{n+1}\bs\Ld_{n+1-\ell},b)} \right. \\
	\left. + 3\tilde{c}^3\frac{C(\Ld_{n+1}\bs\Ld_n,a) D^2(\Ld_{n+1-\ell})C^2(\Ld\bs\Ld_{n+1-\ell},b)}{C^2(\Ld_n,a)C(\Ld_{n+1}\bs\Ld_{n+1-\ell},a)C^3(\Ld_n,b)C(\Ld_{n+1}\bs\Ld_{n+1-\ell},b)}\right]	
\end{align*}
We group ratios, cancel out terms, and bound ratios of terms not needed for later bound using $\Ld' \subseteq \Ld \Rightarrow C(\Ld',s)/C(\Ld,s) \leq 1$.  The terms without $D(\Ld)$ (first, second, fourth, and fifth terms) can be bounded by a ratio of $a$ normalization constants:
\begin{align*}
\frac{C(\Ld_{n+1-\ell},a)}{C(\Ld_n,a)}\frac{C(\Ld_{n+1}\bs\Ld_n,a)}{C(\Ld_n,a)}.
\end{align*}
which we will group together (and multiply by $\tilde{c}\geq 1$ to have the same power of $\tilde{c}$).  
The third term will be left alone.  
The sixth and seventh terms require more care.  
We will use the bound $\Ld \supseteq \Ld'\ \Rightarrow\ C(\Ld,s)\geq C(\Ld',s)$ in the denominator of the sixth term.
For the seventh term, we will use the bound $D(\Ld_{n+1-\ell}) \leq C(\Ld_{n+1-\ell},a)C(\Ld_{n+1-\ell},b)$.  
These steps lead to the upper bound
\begin{align*}
	\leq 3\sum |c_{b,a}|^2\left[ 7 \tilde{c}^3\frac{C(\Ld_{n+1-\ell},a)}{C(\Ld_n,a)}\frac{C(\Ld_{n+1}\bs\Ld_n,a)}{C(\Ld_n,a)} \right.\\
	 + 2\tilde{c}^2\frac{C(\Ld_{n+1}\bs\Ld_n,a)}{C(\Ld_{n+1}\bs\Ld_{n+1-\ell},a)}\frac{C(\Ld_{n+1-\ell},b)}{C(\Ld_n,b)}\frac{D^2(\Ld_n)}{C^2(\Ld_n,a)C^2(\Ld_n,b)}\\
	 + 3\tilde{c}^3\frac{C(\Ld_{n+1}\bs\Ld_n,a)}{C(\Ld_{n+1}\bs\Ld_{n+1-\ell},a)}\frac{C^2(\Ld_{n+1-\ell},b)}{C^2(\Ld_n,b)}\frac{D^2(\Ld_n\bs\Ld_{n+1-\ell})}{C^2(\Ld_n\bs\Ld_{n+1-\ell},a)C^2(\Ld_n\bs\Ld_{n+1-\ell},b)} \\
	 \left. + 3\tilde{c}^3\frac{C(\Ld_{n+1}\bs\Ld_n,a)}{C(\Ld_{n+1}\bs\Ld_{n+1-\ell},a)}\frac{C^2(\Ld_{n+1-\ell},a)}{C^2(\Ld_n,a)}\frac{C^2(\Ld_{n+1-\ell},b)}{C^2(\Ld_n,b)}\right].
\end{align*}

We use the ratio lemmas to show the above bound is nearly exponentially small in $\ell$.  
The ratios of normalization constants in the first, second, and fourth terms are bounded above by
\begin{align}
\frac{C(\Ld_{n+1-\ell},a)}{C(\Ld_n,a)}\frac{C(\Ld_{n+1}\bs\Ld_n,a)}{C(\Ld_n,a)} \leq \exp\left(-2(\ell-2)|\log\tilde{\ld}_a|\right) \label{secondbound}
\end{align}
which follows from inspection by cases.  
If $\tilde{\ld}_a >1$ the first ratio is bounded above by $\exp\left(-2(\ell-1)|\log\tilde{\ld}_a|\right)$ \ref{4R1} and the second by $\exp(2|\log\tilde{\ld}_a|)$ \ref{4R3}.
If $\tilde{\ld}_a <1$ the first ratio is bounded by one \ref{4L1} and second by $\exp\left(-2(n+1)|\log\tilde{\ld}_a|\right)$ \ref{4L3} and the bound follows because $n\geq \ell$ .

The normalization constants in the third term are bounded above by 
\begin{align*}
	\frac{C(\Ld_{n+1}\bs\Ld_n,a)}{C(\Ld_{n+1}\bs\Ld_{n+1-\ell},a)}\frac{C(\Ld_{n+1-\ell},b)}{C(\Ld_n,b)} \frac{D^2(\Ld_n)}{C^2(\Ld_n,a)C^2(\Ld_n,b)}\\ \leq (\ell-1)\exp\left(-2(\ell-2)|\log\tilde{\ld}_a|\right) + \exp\left(-2(\ell-1)|\log\tilde{\ld}_b|\right)\\
	\leq 2\ell\max\left\{ \left(-2(\ell-1)|\log\tilde{\ld}_a|\right), \left(-2(\ell-1)|\log\tilde{\ld}_b|\right)\right\} 
\end{align*}
which follows from cases.  
Each of these ratios are bounded above by 1.  
If $\tilde{\ld}_a <1$, the first ratio is bounded above by $\exp\left(-2(\ell-1)|\log\tilde{\ld}_a|\right)$ \ref{4L2}.  
If $\tilde{\ld}_b >1$, the second ratio is bounded above by $\exp\left(-2(\ell-1)|\log\tilde{\ld}_b|\right)$ \ref{4R1}.  
If $\tilde{\ld}_a>1$ and $\tilde{\ld}_b<1$, then the third ratio is bounded above by $n\left(\exp\left(-2(n-1)|\log\tilde{\ld}_a|\right)+ \exp\left(-2(n-1)|\log\tilde{\ld}_b|\right)\right)$ from diagonal bound \ref{diagonalbound} which gives the appropriate bound since $n\geq \ell$ and $\ell$ is larger than $2+(|\log\tilde{\ld}_a|)^{-1}$ which is where the maximum is achieved in the expression above.

The ratios in the sixth and seventh terms are grouped as  
\begin{align*}
\frac{C(\Ld_{n+1}\bs\Ld_n,a)}{C(\Ld_{n+1}\bs\Ld_{n+1-\ell},a)}\frac{D^2(\Ld_n\bs\Ld_{n+1-\ell})}{C^2(\Ld_n,a)C(\Ld_n,b)C(\Ld_{n+1}\bs\Ld_{n+1-\ell},b)}\frac{C^2(\Ld_{n+1-\ell},b)}{C^2(\Ld_n,b)}\\
	+\frac{C(\Ld_{n+1}\bs\Ld_n,a)}{C(\Ld_{n+1}\bs\Ld_{n+1-\ell},a)}\frac{ D^2(\Ld_{n+1-\ell})}{C^2(\Ld_n,a)C^2(\Ld_n,b)}
\end{align*}
where ratios that are not needed for the exponential bound are bounded above by 1 and are dropped.  
Note that the terms in the brackets are bounded above by 1.  
  
If $\tilde{\ld}_a <1$, then the first ratio is both terms less that $\exp\left(-2(\ell-1)|\log\tilde{\ld}_a|\right)$ \ref{4R1} and the sum of other terms is bounded above by 2.
We have the bound 
$$2\exp\left(-2(\ell-1)|\log\tilde{\ld}_a|\right).$$
Suppose $\tilde{\ld}_a>1$.  
The first ratio in each term is bounded by 1 \ref{4L1} and the exponential bound is derived for terms inside the brackets.
For the rest of the first term, 
$$\frac{D^2(\Ld_n\bs\Ld_{n+1-\ell})}{C^2(\Ld_n,a)C(\Ld_n,b)C(\Ld_{n+1}\bs\Ld_{n+1-\ell},b)}\frac{C^2(\Ld_{n+1-\ell},b)}{C^2(\Ld_n,b)},$$
if $\tilde{\ld}_b>1$, the first ratio is bounded above by 1 and the second by $\exp\left(-2(\ell-1)|\log\tilde{\ld}_b|\right)$ \ref{4R1}.  If $\tilde{\ld}_b <1$, the second ratio is bounded above by 1 and the first ratio is bounded by
\begin{align*}
\frac{D^2(\Ld_n\bs\Ld_{n+1-\ell})}{C^2(\Ld_n,a)C(\Ld_n,b)C(\Ld_{n+1}\bs\Ld_{n+1-\ell},b)} \leq \frac{D(\Ld_n\bs\Ld_{n+1-\ell})}{C(\Ld_n\bs\Ld_{n+1-\ell},a)C(\Ld_n\bs\Ld_{n+1-\ell},b)}\\
\leq \ell\exp\left(-2(\ell-1)|\log\tilde{\ld}_a|\right) + \exp\left(-2(\ell-1)|\log\tilde{\ld}_b|\right)\\
\leq 2\ell \max\left\{ \left(-2(\ell-1)|\log\tilde{\ld}_a|\right), \left(-2(\ell-1)|\log\tilde{\ld}_b|\right)\right\} 
\end{align*}
For the second term , we use the upper bound $D(\Ld) \leq C(\Ld,a)C(\Ld,b)$ and $\tilde{\ld}_a >1$, to obtain the upper bound
\begin{align*}
	\frac{ D^2(\Ld_{n+1-\ell})}{C^2(\Ld_n,a)C^2(\Ld_n,b)} \leq \frac{ C^2(\Ld_{n+1-\ell},a) C^2(\Ld_{n+1-\ell},b)}{C^2(\Ld_n,a)C^2(\Ld_n,b)}\\
	\leq \exp\left(-2(\ell-1)|\log\tilde{\ld}_a|\right)
\end{align*}
where the $b$ terms are bounded above by one.
Therefore, the bound on the sixth and seventh terms is 
$$3\ell \max\left\{ \left(-2(\ell-1)|\log\tilde{\ld}_a|\right), \left(-2(\ell-1)|\log\tilde{\ld}_b|\right)\right\}$$

The $c_{b,a}$ terms are bounded above by
\begin{align}\label{abbound}
	60\ell\tilde{c}^3\max\left\{ \left(-2(\ell-2)|\log\tilde{\ld}_a|\right), \left(-2(\ell-2)|\log\tilde{\ld}_b|\right)\right\} \sum |c_{b,a}|^2
\end{align}
and the operator norm squared restricted to the $b,a$ subspace is bounded above by
$$60\ell\tilde{c}^3\max\left\{ \left(-2(\ell-1)|\log\tilde{\ld}_a|\right), \left(-2(\ell-1)|\log\tilde{\ld}_b|\right)\right\}$$
By interchanging $a$ and $b$, we have a similar bound on $c_{a,b}$ terms.

Finally, we bound the $c_{0,ab}$ terms:
\begin{align*}
	3\left|\sum c_{0,ab}\tilde{\ld}_a^x\tilde{\ld}_b^y\right|^2\left[\frac{C(\Ld_{n+1-\ell},ab)}{C^2(\Ld_n,ab)} \right. + \frac{C(\Ld_{n+1-\ell},b)C^2(\Ld_n\bs\Ld_{n+1-\ell},a)}{C(\Ld_{n+1}\bs\Ld_{n+1-\ell},a)C^2(\Ld_n,ab)} \nonumber\\
	+\left.  \frac{C(\Ld_{n+1-\ell},a)C^2(\Ld_n\bs\Ld_{n+1-\ell},b)}{C(\Ld_{n+1}\bs\Ld_{n+1-\ell},b)C^2(\Ld_n,ab)} + \frac{\left|C(\Ld_n,ab)-C(\Ld_n\bs\Ld_{n+1-\ell},ab)\right|^2}{C(\Ld_{n+1}\bs\Ld_{n+1-\ell},ab)C^2(\Ld_n,ab)}\right]\\
	 \leq 3\sum|c_{0,ab}|^2\left[ \tilde{c}^2\frac{C(\Ld_{n+1}\bs\Ld_n,a)C(\Ld_{n+1-\ell},a)}{C^2(\Ld_n,a)}\frac{C(\Ld_{n+1}\bs\Ld_n,b)C(\Ld_{n+1-\ell},b)}{C^2(\Ld_n,b)} \right. \\
	+\tilde{c}^2\frac{C(\Ld_{n+1}\bs\Ld_n,a)C^2(\Ld_n\bs\Ld_{n+1-\ell},a)}{C^2(\Ld_n,a)C(\Ld_{n+1}\bs\Ld_{n+1-\ell},a)}\frac{C(\Ld_{n+1}\bs\Ld_n,b)C(\Ld_{n+1-\ell},b)}{C^2(\Ld_n,b)}\\
		+\tilde{c}^2\frac{C(\Ld_{n+1}\bs\Ld_n,a)C(\Ld_{n+1-\ell},a)}{C^2(\Ld_n,a)}\frac{C(\Ld_{n+1}\bs\Ld_n,b)C^2(\Ld_n\bs\Ld_{n+1-\ell},b)}{C^2(\Ld_n,b)C(\Ld_{n+1}\bs\Ld_{n+1-\ell},b)}\\
		+3\tilde{c}^3\frac{C(\Ld_{n+1}\bs\Ld_n,a)C^2(\Ld_{n+1-\ell},a)}{C(\Ld_{n+1}\bs\Ld_{n+1-\ell},a)C^2(\Ld_n,a)} \frac{C(\Ld_{n+1}\bs\Ld_n,b)C^2(\Ld_{n+1-\ell},b)}{C(\Ld_{n+1}\bs\Ld_{n+1-\ell},b)C^2(\Ld_n,b)}\\
		+3\tilde{c}^3\frac{C(\Ld_{n+1}\bs\Ld_n,a)C^2(\Ld_n\bs\Ld_{n+1-\ell},a)}{C(\Ld_{n+1}\bs\Ld_{n+1-\ell},a)C^2(\Ld_n,a)}\frac{C(\Ld_{n+1}\bs\Ld_n,b)C^2(\Ld_{n+1-\ell},b)}{C(\Ld_{n+1}\bs\Ld_{n+1-\ell},b)C^2(\Ld_n,b)}\\	
		\left. +3\tilde{c}^3\frac{C(\Ld_{n+1}\bs\Ld_n,a)C^2(\Ld_{n+1-\ell},a)}{C(\Ld_{n+1}\bs\Ld_{n+1-\ell},a)C^2(\Ld_n,a)}\frac{C(\Ld_{n+1}\bs\Ld_n,b)C^2(\Ld_n\bs\Ld_{n+1-\ell},b)}{C(\Ld_{n+1}\bs\Ld_{n+1-\ell},b)C^2(\Ld_n,b)}\right]
\end{align*}
The first, third, fourth, and sixth terms can be bounded exponentially by the $a$ ratio
\begin{align*}
	\frac{C(\Ld_{n+1}\bs\Ld_n,a)C(\Ld_{n+1-\ell},a)}{C^2(\Ld_n,a)} \leq \exp\left(-2(\ell-2)|\log\tilde{\ld}_a|\right)
\end{align*}
by \ref{secondbound}.  
The first, second, fourth and fifth terms can be bounded exponentially by the $b$ ratios of normalization constants:
\begin{align*}
	\frac{C(\Ld_{n+1}\bs\Ld_n,b)C(\Ld_{n+1-\ell},b)}{C^2(\Ld_n,b)} \leq \exp\left(-2(\ell-2)|\log\tilde{\ld}_b|\right)
\end{align*}
The sum of all these terms are bounded above by
\begin{align}\label{0abbound}
36\tilde{c}^3\max\left\{ \left(-2(\ell-2)|\log\tilde{\ld}_a|\right), \left(-2(\ell-2)|\log\tilde{\ld}_b|\right)\right\} \sum |c_{0,ab}|^2
\end{align}
To combine the bounds, we take the maximum of all the operator norm-squared projection bounds on each subspace which is 
$$60\ell\tilde{c}^3\max\left\{ \left(-2(\ell-2)|\log\tilde{\ld}_a|\right), \left(-2(\ell-2)|\log\tilde{\ld}_b|\right)\right\}$$
Therefore,
\begin{align}
	\|G_nE_n\| \leq \sqrt{(60\ell)}\tilde{c}^{3/2} \max\left\{ \left(-(\ell-2)|\log\tilde{\ld}_{a,j}|\right), \left(-(\ell-2)|\log\tilde{\ld}_{b,j}|\right)\right\}
\end{align}
when $\ell \geq \max\{(\log\tilde{\ld}_{a,j})^{-1}, (\log\tilde{\ld}_{b,j})^{-1}\}$.


\bibliographystyle{plain}
\bibliography{mybib}   

\begin{thebibliography}{10}

\bibitem{aklt88}
I.~Affleck, T.~Kennedy, E.~H. Lieb, and H.~Tasaki.
\newblock Valence bond ground states in isotropic quantum antiferromagnets.
\newblock In {\em Condensed Matter Physics and Exactly Soluble Models}, pages
  253--304. Springer, 1988.

\bibitem{Aharonov:2009:DLQ:1536414.1536472}
D.~Aharonov, I.~Arad, Z.~Landau, and U.~Vazirani.
\newblock The detectability lemma and quantum gap amplification.
\newblock In {\em Proceedings of the Forty-first Annual ACM Symposium on Theory
  of Computing}, STOC '09, pages 417--426, New York, NY, USA, 2009. ACM.

\bibitem{aharonov2011detectability}
D.~Aharonov, I.~Arad, U.~Vazirani, and Z.~Landau.
\newblock The detectability lemma and its applications to quantum hamiltonian
  complexity.
\newblock {\em New journal of physics}, 13(11):113043, 2011.

\bibitem{PhysRevB.93.205142}
A.~Anshu, I.~Arad, and T.~Vidick.
\newblock Simple proof of the detectability lemma and spectral gap
  amplification.
\newblock {\em Phys. Rev. B}, 93:205142, May 2016.

\bibitem{arad2013area}
Itai Arad, Alexei Kitaev, Zeph Landau, and Umesh Vazirani.
\newblock An area law and sub-exponential algorithm for 1d systems.
\newblock {\em arXiv preprint arXiv:1301.1162}, 2013.

\bibitem{BHNY2015}
S.~Bachmann, E.~Hamza, B.~Nachtergaele, and A.~Young.
\newblock Product vacua and boundary state models in d-dimensions.
\newblock {\em Journal of Statistical Physics}, 160(3):636--658, 2015.

\bibitem{BMNS2012}
S.~Bachmann, S.~Michalakis, B.~Nachtergaele, and R.~Sims.
\newblock Automorphic equivalence within gapped phases of quantum lattice
  systems.
\newblock {\em Comm. Math. Phys.}, 309(3):835--871, 2012.

\bibitem{bachmann2014product}
S.~Bachmann and B.~Nachtergaele.
\newblock Product vacua with boundary states and the classification of gapped
  phases.
\newblock {\em Communications in Mathematical Physics}, 329(2), 2014.

\bibitem{bishop2016spectral}
M.~Bishop, B.~Nachtergaele, and A.~Young.
\newblock Spectral gap and edge excitations of d-dimensional pvbs models on
  half-spaces.
\newblock {\em Journal of Statistical Physics}, 162(6):1485--1521, 2016.

\bibitem{PhysRevLett.109.207202}
S.~Bravyi, L.~Caha, R.~Movassagh, D.~Nagaj, and P.~W. Shor.
\newblock Criticality without frustration for quantum spin-1 chains.
\newblock {\em Phys. Rev. Lett.}, 109:207202, Nov 2012.

\bibitem{2015JMP....56f1902B}
S.~{Bravyi} and D.~{Gosset}.
\newblock {Gapped and gapless phases of frustration-free spin- /1 2 chains}.
\newblock {\em Journal of Mathematical Physics}, 56(6):061902, June 2015.

\bibitem{PhysRevB.83.035107}
X.~Chen, Z.C. Gu, Z.X. Liu, and X.G. Wen.
\newblock Classification of gapped symmetric phases in one-dimensional spin
  systems.
\newblock {\em Phys. Rev. B}, 83:035107, Jan 2011.

\bibitem{PhysRevB.87.155114}
X.~Chen, Z.C. Gu, Z.X. Liu, and X.G. Wen.
\newblock Symmetry protected topological orders and the group cohomology of
  their symmetry group.
\newblock {\em Phys. Rev. B}, 87:155114, Apr 2013.

\bibitem{cubitt2015undecidability}
T.S. Cubitt, D.~Perez-Garcia, and M.~Wolf.
\newblock Undecidability of the spectral gap.
\newblock {\em Nature}, 528(7581):207--211, 2015.

\bibitem{fannes1992}
M.~Fannes, B.~Nachtergaele, and R.~F. Werner.
\newblock Finitely correlated states on quantum spin chains.
\newblock {\em Comm. Math. Phys.}, 144(3):443--490, 1992.

\bibitem{farhi2000quantum}
Edward Farhi, Jeffrey Goldstone, Sam Gutmann, and Michael Sipser.
\newblock Quantum computation by adiabatic evolution.
\newblock {\em arXiv preprint quant-ph/0001106}, 2000.

\bibitem{PhysRevLett.116.097202}
D.~Gosset and Y.~Huang.
\newblock Correlation length versus gap in frustration-free systems.
\newblock {\em Phys. Rev. Lett.}, 116:097202, Mar 2016.

\bibitem{gosset2016local}
D.~Gosset and E.~Mozgunov.
\newblock Local gap threshold for frustration-free spin systems.
\newblock {\em Journal of Mathematical Physics}, 57(9):091901, 2016.

\bibitem{gottesman2010entanglement}
D.~Gottesman and M.~B. Hastings.
\newblock Entanglement versus gap for one-dimensional spin systems.
\newblock {\em New journal of physics}, 12(2):025002, 2010.

\bibitem{PhysRevB.69.104431}
M.~B. Hastings.
\newblock Lieb-schultz-mattis in higher dimensions.
\newblock {\em Phys. Rev. B}, 69:104431, Mar 2004.

\bibitem{Hastings07}
M.~B. Hastings.
\newblock Quasi-adiabatic continuation in gapped spin and fermion systems:
  Goldstone's theorem and flux periodicity.
\newblock {\em Journal of Statistical Mechanics: Theory and Experiment},
  2007(05):P05010, 2007.

\bibitem{Hastings2006}
M.B. Hastings and T.~Koma.
\newblock Spectral gap and exponential decay of correlations.
\newblock {\em Communications in Mathematical Physics}, 265(3):781--804, 2006.

\bibitem{irani2010ground}
S.~Irani.
\newblock Ground state entanglement in one-dimensional translationally
  invariant quantum systems.
\newblock {\em Journal of Mathematical Physics}, 51(2):022101, 2010.

\bibitem{jansen2007bounds}
Sabine Jansen, Mary-Beth Ruskai, and Ruedi Seiler.
\newblock Bounds for the adiabatic approximation with applications to quantum
  computation.
\newblock {\em Journal of Mathematical Physics}, 48(10):102111, 2007.

\bibitem{kastoryano2017divide}
Michael~J Kastoryano and Angelo Lucia.
\newblock Divide and conquer method for proving gaps of frustration free
  hamiltonians.
\newblock {\em arXiv preprint arXiv:1705.09491}, 2017.

\bibitem{kitaev2003fault}
A.~Y. Kitaev.
\newblock Fault-tolerant quantum computation by anyons.
\newblock {\em Annals of Physics}, 303(1):2--30, 2003.

\bibitem{knabe1988energy}
S.~Knabe.
\newblock Energy gaps and elementary excitations for certain vbs-quantum
  antiferromagnets.
\newblock {\em Journal of statistical physics}, 52(3-4):627--638, 1988.

\bibitem{lemm2018spectral}
Marius Lemm and Evgeny Mozgunov.
\newblock Spectral gaps of frustration-free spin systems with boundary.
\newblock {\em arXiv preprint arXiv:1801.08915}, 2018.

\bibitem{movassagh2014power}
R.~Movassagh and P.~W. Shor.
\newblock Power law violation of the area law in quantum spin chains.
\newblock {\em arXiv preprint arXiv:1408.1657}, 2014.

\bibitem{movassagh2017generic}
Ramis Movassagh.
\newblock Generic local hamiltonians are gapless.
\newblock {\em Physical review letters}, 119(22):220504, 2017.

\bibitem{nachtergaele96}
B.~Nachtergaele.
\newblock The spectral gap for some spin chains with discrete symmetry
  breaking.
\newblock {\em Communications in Mathematical Physics}, 175(3):565--606, 1996.

\bibitem{Nachtergaele2006}
B.~Nachtergaele and R.~Sims.
\newblock Lieb-robinson bounds and the exponential clustering theorem.
\newblock {\em Communications in Mathematical Physics}, 265(1):119--130, 2006.

\bibitem{PhysRevB.84.165139}
N.~Schuch, D.~P\'erez-Garc\'{i}a, and I.~Cirac.
\newblock Classifying quantum phases using matrix product states and projected
  entangled pair states.
\newblock {\em Phys. Rev. B}, 84:165139, Oct 2011.

\bibitem{schwarz2013preparing}
Martin Schwarz, Kristan Temme, Frank Verstraete, David Perez-Garcia, and Toby~S
  Cubitt.
\newblock Preparing topological projected entangled pair states on a quantum
  computer.
\newblock {\em Physical Review A}, 88(3):032321, 2013.

\end{thebibliography}

\end{document}